%% file: rrpq_arxiv.tex
\documentclass[12pt,reqno]{article}
\usepackage[usenames,dvipsnames]{xcolor}

\usepackage{amsmath,mathtools,tikz,pgfplots,soul}

\input{working_arxiv}

\usepackage{bibunits}

\makeatletter
\newcommand*{\setbibunitprefix}[1]{%
  \def\@extra@binfo{@#1}%
  \def\@extra@b@citeb{@#1}%
}
\makeatother

\usepackage[T1]{fontenc}

\usepackage{newtxtext}
\usepackage{lmodern,newtxtext,newtxmath}
\usepackage[cal=cm]{mathalfa}

\usepackage{graphicx} 
\usepackage{amsmath}  
\usepackage{scalerel} 

\makeatletter
\renewcommand{\paragraph}{%
  \@startsection{paragraph}{4}{\z@}%
    {2.5ex \@plus .5ex \@minus .2ex}
    {-1em}
    {\normalfont\normalsize\bfseries}%
}
\makeatother

\usepackage{booktabs,array,tablefootnote,subcaption} 
\usepackage[titletoc]{appendix}
\newcolumntype{C}[1]{>{\centering\arraybackslash}m{#1}}

\makeatletter
\newcommand{\vast}{\bBigg@{4}}
\newcommand{\Vast}{\bBigg@{5}}
\makeatother

\let\Pr\relax
\DeclareMathOperator{\Pr}{Pr}
\DeclareMathOperator*{\esssup}{ess\,sup}
\DeclareMathOperator*{\essinf}{ess\,inf}
\DeclareMathOperator{\LF}{LF}

\definecolor{azure(colorwheel)}{rgb}{0.0, 0.5, 1.0}
\definecolor{brandeisblue}{rgb}{0.0, 0.44, 1.0}
\definecolor{ceruleanblue}{rgb}{0.16, 0.32, 0.75}
\definecolor{airforceblue}{rgb}{0.36, 0.54, 0.66}
\definecolor{bleudefrance}{rgb}{0.19, 0.55, 0.91}
\definecolor{darkspringgreen}{rgb}{0.09, 0.45, 0.27}
\definecolor{rulecolor}{rgb}{0.13, 0.35, 0.5}

\usepgfplotslibrary{fillbetween,decorations.softclip}
\usetikzlibrary{arrows.meta}
\usetikzlibrary{patterns,calc}

\newcommand{\papertitle}{{\bf Prices vs.~Quantities: Robust Regulation}\footnote[$*$]{\thanktext}
}

\newcommand{\name}{Zi Yang Kang\footnote[$\dag$]{\affiliationi}}

\newcommand{\affiliationi}{Department of Economics, University of Toronto; \href{mailto:zy.kang@utoronto.ca}{\tt zy.kang@utoronto.ca}.}

\newcommand{\thanktext}{This paper is dedicated to Jeremy Bulow, whose way of thinking about economics continues to inspire me.  I thank Gabriel Carroll, Rob McMillan, Paul Milgrom, Andy Skrzypacz, Jakub Steiner, Colin Stewart, Shosh Vasserman, and especially Laura Doval, for helpful discussions.  {\tt Refine.ink} was used to check the paper for consistency and clarity.
}
\newcommand{\abstracttext}{
\noindent This paper revisits the classic instrument choice problem in a setting with consumption externalities, through the lens of robust mechanism design.  A regulator chooses among all nonlinear pricing schedules to maximize worst-case welfare, knowing the distribution of consumption preferences and the average marginal externality but not the underlying joint distribution.  The optimal policy is a quantity control: a floor for positive externalities and a ceiling for negative externalities.  Knowing whether individuals with stronger consumption preferences tend to generate larger or smaller marginal externalities can instead make a uniform tax or subsidy optimal.  The framework therefore provides a welfare-based rationale for simple forms of price and quantity regulation.
\\[2ex]
{\em JEL classification}: D81, D82, H23, L51\\[1ex]
{\em Keywords}: quantity regulation, externalities, mechanism design, selection uncertainty, robustness}
\newcommand{\inserttitle}{\begin{center}\hfill\\[0ex]
\LARGE\papertitle\\[3ex]
\large
\name\\[3ex]
{\monthyeardate\today}
\\[4ex]
\end{center}}	

\begin{document}
\onehalfspacing
\inserttitle
\begin{abstract}
\abstracttext
\end{abstract}
\setstretch{1.5}
\thispagestyle{empty}

\clearpage

\setcounter{page}{1}

\begin{bibunit}[econ-econometrica]

\section{Introduction}
\label{sec:introduction}

Policymakers frequently rely on quantity controls to regulate externalities.  For example, the recreational use of some opiates and opioids is widely banned, while immunizations are often mandated by schools and employers.  Even when outright bans and mandates are not used, quantity restrictions may still apply.  Most countries impose compulsory education laws that require children to receive some minimum amount of schooling.  One-gun-a-month laws and the Combat Methamphetamine Epidemic Act, respectively, limit the quantity of firearms and pseudoephedrine that any individual may purchase within a given month.

However, many economists believe in the superiority of price regulation through taxes and subsidies.  The COVID-19 pandemic provides a stark illustration of this tension between economic prescriptions and policy practice.  Several prominent economists, including Robert Litan and Gregory Mankiw, advocated paying individuals to get vaccinated in order to achieve herd immunity \citep{litan20,mankiw20}.  Yet few schools, employers, or governments heeded these calls, choosing instead to impose vaccine mandates.

At first glance, this may look like the familiar prices-versus-quantities problem studied by \citet{weitzman74}.  In that setting, the regulator chooses between a fixed price and a fixed quantity under Bayesian uncertainty about shifts in the marginal benefit and marginal cost curves; their relative slopes determine which instrument performs better.  However, as a large subsequent literature has emphasized, two aspects of that benchmark are ill-suited to many policy settings.  First, it places substantial informational demands on the regulator \citep{baumoloates88}.  Second, it restricts attention to simple instruments, even though sufficiently rich price regulation (e.g., nonlinear pricing) can dominate quantity controls \citep{kaplowshavell02}.  Put differently, \citeauthor{weitzman74}'s approach is most natural in settings with relatively \emph{rich} information but a \emph{limited} set of instruments, whereas many real-world environments feature \emph{limited} information even when a \emph{rich} set of instruments is feasible.

This paper takes a new approach to the instrument choice problem by embedding it in a mechanism-design framework with limited information and a rich set of feasible instruments.  This framework comprises two key ingredients: \emph{heterogeneity} and \emph{robustness}.  Agents are \emph{heterogeneous} in that they differ both in their consumption preferences for a homogeneous good and in the amount of externality they generate per unit of consumption.  In the COVID-19 vaccination example, individuals vary in their willingness to pay for the vaccine as well as in the marginal social benefit their vaccination generates, for instance because they differ in their propensity to wear masks or avoid large crowds.  While the regulator may know the distribution of consumption preferences---which can be estimated from observed consumption decisions---she does not observe the amount of externality generated by each agent and may have only limited information about its distribution, such as its mean.  This reflects real-world settings in which policymakers know the average benefit of vaccination, but not how marginal benefits are distributed or depend on willingness to pay.  Faced with this uncertainty, the regulator adopts a \emph{robust} approach by evaluating policies according to their worst-case welfare performance.  Importantly, the regulator is not restricted to simple instruments \emph{a priori}: she can choose any nonlinear pricing schedule.

The first main result shows that robustness to uncertainty about how marginal benefits depend on consumption preferences favors quantity regulation.  Although the regulator can choose any nonlinear pricing schedule, her unique optimal allocation can be implemented by a simple policy: imposing a quantity floor or ceiling when the externality is positive or negative, respectively.  In special cases, such as when agents have unit demand, this result takes the form of mandates and bans.

To see why, consider the COVID-19 vaccination example.  On the one hand, paying individuals to get vaccinated performs well when those with the highest willingness to pay also generate the largest marginal social benefit.  On the other hand, if individuals who are most eager to get vaccinated already take extensive precautions, then marginal benefit may be concentrated among individuals with the lowest willingness to pay.  Any subsidy that leaves a positive mass of individuals unvaccinated can therefore generate essentially no external benefit in the worst case.  By contrast, a mandate guarantees full participation, albeit at the cost of potentially vaccinating some individuals whose willingness to pay and marginal social benefit are both low.  Depending on how this tradeoff is resolved, an optimal policy is either to mandate vaccination or not to intervene at all.

The second main result shows that the optimal regulatory instrument---price or quantity---depends critically on the direction of \emph{selection} between consumption preferences and marginal external benefits.  Intuitively, selection is positive when marginal external benefits tend to be larger among agents with stronger consumption preferences, and negative when they tend to be smaller.  In particular, information about the direction of selection can rule out the worst-case selection patterns that drive the case for quantity regulation when selection is unknown.  Under negative selection, the previous logic continues to apply and quantity regulation remains optimal.  However, under positive selection---for instance, if individuals who are most eager to get vaccinated are also those at greatest risk of exposure and transmission---price regulation becomes optimal.  In this case, the optimal allocation can be implemented by a uniform subsidy equal to the \emph{average marginal external benefit}.  Whether selection is positive is unclear in the COVID-19 context, but it is more plausible for other vaccines, such as those for human papillomavirus (HPV) or mpox, where individuals at higher risk may also have higher willingness to pay.  Similar considerations arise in other externality settings: for example, many countries subsidize electric vehicle adoption, where individuals with higher willingness to pay, such as high-mileage drivers, are expected to generate larger marginal social benefits through greater emissions reductions.

It is worth emphasizing that both heterogeneity and robustness are crucial for the results in this paper.  Without heterogeneity, the marginal amount of externality generated by each agent would be the same.  In this case, the classical analysis due to \citet{pigou20} applies: the regulator can do no better than to set a tax or subsidy equal to that amount of externality.  Without robustness, the regulator might instead maximize expected welfare in a standard Bayesian environment with a known joint distribution of willingness to pay and amount of externality generated.  This is a substantial informational requirement on the regulator---a point that I revisit below.  But putting that aside, the optimal policy would generally depend on how marginal externalities vary with consumption preferences and need not take a simple form \citep{kaplowshavell02}.  Instead, as I show below, the robustness criterion determines a least-favorable pattern of selection (or a limiting sequence thereof) that does not depend on the regulator's policy.  In this sense, it is the interaction between heterogeneity and robustness that drives the results in this paper.

These results extend in several directions.  Revenue concerns change the intensity---but not the form---of optimal regulation.  When marginal external benefits have known bounded support, the optimal policy under unknown selection takes the form of a two-tier subsidy; for a sufficiently large upper support bound, this policy implements a quantity floor.  Finally, the Supplemental Appendix shows that the same robustness logic applies beyond heterogeneous consumption externalities.  Under regulatory uncertainty \emph{\`a la} \citet{weitzman74}, it favors a quantity control even when the regulator may choose any nonlinear payment schedule; in environments with socially wasteful screening, it favors random allocation.

\subsection{Related Literature}

This paper shares a similar motivation with a positive literature that studies why quantity regulation is so prevalent even when price regulation is theoretically appealing.  Key contributions by \citet{glaesershleifer01}, \citet{glaeseretal01}, and \citet{glaesershleifer02} develop two complementary enforcement-based explanations: quantity restrictions can both strengthen enforcement incentives and reduce enforcement costs by making violations easier to detect.  A different strand of the literature, pioneered by \citet{buchanantullock75}, examines political economy mechanisms.  The present paper complements these explanations by developing a welfare-based argument for quantity regulation and by adding nuance: price regulation---also observed in practice---can be explained by the same argument when the regulator has information about the direction of selection between consumption preferences and marginal externality.

The argument developed in this paper relates directly to a growing literature that emphasizes the importance of heterogeneity in the regulation of externalities.  The seminal work of \citet{diamond73} shows that when individuals generate heterogeneous marginal externalities, the optimal uniform corrective tax generally differs from the ``na\"ive'' uniform tax based on the unweighted average marginal harm, and instead depends on how marginal externalities vary with individuals' demand elasticities.  A recent empirical literature has documented such heterogeneity and selection in specific settings where rich data are available, including work by \citet{knittelsandler18}, \citet{griffithetal19}, and \citet{grummonetal19}.  At the same time, these studies underscore how demanding the informational requirements of optimal taxation are: identifying the relevant correlation structure typically requires detailed micro-level data on both consumption and social costs, data that are unavailable in many policy settings in which such heterogeneity is likely to be quantitatively important, such as those discussed by \citet{fleischer15}.  In this sense, the present paper complements the existing literature by studying robust regulatory policies when such detailed data are unavailable.

Another closely related literature studies the normative question of how to choose between price and quantity instruments in regulation.  A key premise of this literature that resonates with the present paper is the difficulty of measuring externalities precisely, as described by \citet{baumoloates71}:
\begin{quote}
``The proper level of the Pigouvian tax (subsidy) [\dots] is equal to the marginal net damage (benefit) produced by that activity.  The difficulty is that it is usually not easy to obtain a reasonable estimate of the money value of this marginal damage.''
\end{quote}
Given this difficulty, pioneering contributions by \citet{weitzman74}, \citet{robertsspence76}, and \citet{baumoloates88} examine the relative performance of price and quantity regulation, as well as hybrid instruments.  Subsequent work has refined these insights and applied them to important policy environments (\eg, \citealp{pizer02}; \citealp{newellpizer03}).  However, \citet{kaplowshavell02} have argued that these results rely on restrictions on price regulation---such as linear taxes---and that sufficiently rich price instruments can do better.  The present paper departs from this literature in two ways: \emph{(i)}~the regulator can choose any nonlinear pricing schedule; and \emph{(ii)}~the regulator seeks robustness to uncertainty by evaluating policies under a worst-case welfare criterion, with limited information about the joint distribution of consumption preferences and marginal externalities.

The focus on robustness also connects this paper to the literature on mechanism design with worst-case objectives.  \citet{carroll19} provides a detailed survey of this literature.  Recent work has applied worst-case criteria to justify simple policies on the basis of robustness, including \citet{frankel14} in delegation; \citet{garrett14} in cost-based procurement; \citet{carroll15}, \citet{waltoncarroll22}, and \citet{carrollbolte23} in moral hazard; \citet{carroll17} in multidimensional screening; \citet{guoshmaya25} in monopoly regulation; and \citet{kangvasserman25} in welfare analysis.  Closer to the present paper's focus on how worst-case reasoning bears on the comparison between price and quantity regulation, \citet{mishraetal25} study a monopoly regulation problem in which the regulator has a conjectured model but does not fully trust it: the regulator selects mechanisms that maximize a worst-case welfare guarantee over a set of plausible models and chooses among them by their expected performance under the conjectured model.\footnote{\citet{dworczakpavan22} use a similar selection rule in robust information design, while \citet{kambhampati26} develops lexicographic refinements of max--min optimality.}  Their analysis compares price and quantity regulations as rich families of direct mechanisms in which the regulator specifies either a price rule or a quantity rule as a function of the monopolist's reported type.  Because demand is itself uncertain, these families are not equivalent, and the comparison is driven by how a regulated price maps into different realized quantities across plausible demand models.  They show that price and quantity regulations yield the same worst-case welfare guarantee and establish conditions under which one family outperforms the other.  By contrast, the present paper studies a different economic environment and source of uncertainty: the regulator wishes to regulate a heterogeneously generated externality and seeks robustness with respect to the distribution of externalities---in particular, how marginal externality depends on consumption preferences.\footnote{This informational structure (where only mean or moment-based information is available) has been considered by other authors, such as \citet{carrascoetal18} and \citet{chezhong25}, who build on \citepos{bergemannschlag11} analysis of monopoly pricing under worst-case objectives when the seller is uncertain about demand.}  In this setting, although the regulator can choose any nonlinear pricing schedule, the unique optimal allocation in each benchmark can be implemented by a simple policy: price regulation takes the form of a uniform tax or subsidy, and quantity regulation takes the form of a quantity ceiling or floor \emph{without} other instruments.  This delivers a clean instrument-choice implication: when selection is unknown, quantity regulation is optimal, whereas information about the direction of selection can make price regulation optimal.

\section{Framework}
\label{sec:framework}

This section formulates the regulator's problem in the presence of a positive consumption externality and discusses the main modeling assumptions.

\subsection{Model}

There is a unit mass of agents in a market for a homogeneous good.  Consumption of the good generates a positive externality.  Agents differ both in their consumption preferences and in the amount of externality they generate from consumption.  The good is supplied competitively at a constant marginal cost, $c>0$.

Heterogeneity in consumption preferences is captured by a \emph{type}, $\theta\in\Theta\coloneq \prod_{i=1}^d[\und\theta_i,\BAR\theta_i]\subset\R^d$.  Even though consumption is one-dimensional, agents may differ along several dimensions of their consumption preferences.  For example, they may differ both in their willingness to pay for initial units and in how rapidly their marginal utility declines with consumption.  The type space is endowed with the product order, and types are distributed according to a distribution $F$ that admits a positive density $f$ on $\Theta$.  Each agent is privately informed of $\theta$, which determines the utility $u(q,\theta)$ that he derives from consuming a quantity $q$ of the good.  To ensure that consumption is bounded, I assume there exists a maximum quantity $A>0$ that each agent can consume, so that $q\in[0,A]$.  As is standard, $u:[0,A]\times\Theta\to\R$ is assumed to be jointly continuous and twice continuously differentiable in quantity, with $u_q>0$ and $u_{qq}<0$ for $q\in(0,A)$, and to satisfy strict increasing differences in $(q,\theta)$: $u(q_H,\theta_H)-u(q_L,\theta_H)>u(q_H,\theta_L)-u(q_L,\theta_L)$ for every $q_H>q_L$ and $\theta_H\geq\theta_L$ with $\theta_H\ne\theta_L$.

Heterogeneity in externality is captured by $\xi\in\R_+$, which denotes the marginal benefit (measured in dollar terms) that the agent generates per unit of consumption.  For simplicity, agents are also assumed to be privately informed about $\xi$ (as shown below, whether agents observe $\xi$ does not affect the analysis).  

Each agent's utility depends on both his own consumption of the good and the aggregate externality generated by all agents.  For expositional simplicity, I focus on the case in which utility is additively separable and linear in the aggregate externality.  Denote the joint distribution of $(\theta,\xi)$ by $H$, whose marginal distribution over $\theta$ is $F$.  Thus, a type-$\theta$ agent who generates an externality of $\xi$ per unit of consumption, consumes a quantity $q(\theta,\xi)$, and makes a payment of $t(\theta,\xi)$ obtains a utility of
\[u(q(\theta,\xi),\theta) - t(\theta,\xi) + E(q),\quad\text{where }E(q)\coloneq\iint_{(\hat\theta,\hat\xi)\in\Theta\times\R_+}\hat\xi q(\hat\theta,\hat\xi)\,\dd H(\hat\theta,\hat\xi).\]

A regulator seeks to regulate the positive externality generated by consumption.  To this end, she commits to a common pricing schedule $T:[0,A]\to\R\cup\{+\infty\}$, under which an agent consuming $q$ makes a payment $T(q)$ that does not depend on other agents' consumption.\footnote{Although this restriction is not without loss of generality when $H$ is unknown, it reflects the usual implementation of taxes, subsidies, and quantity restrictions, which are based only on each agent's own consumption.}  By the revelation principle \citep{myerson79}, the outcome induced by such a pricing schedule can be represented by an incentive-compatible mechanism $(q,t)$, where $q:\Theta\times\R_+\to[0,A]$ and $t:\Theta\times\R_+\to\R$ respectively determine an agent's quantity and payment as functions of his reported type.  Conversely, it is easy to see that any such incentive-compatible mechanism can be implemented by a common pricing schedule.

Because an individual report does not affect the aggregate externality $E(q)$, incentive compatibility requires that, for every $(\theta,\xi)\in\Theta\times\R_+$,
\begin{equation}
(\theta,\xi)\in\argmax_{(\hat\theta,\hat\xi)\in\Theta\times\R_+}\left[u(q(\hat\theta,\hat\xi),\theta) - t(\hat\theta,\hat\xi)\right]. \tag{IC}\label{eq:IC}
\end{equation}
Standard arguments imply that $(q,t)$ must be essentially independent of $\xi$: as can be seen from \eqref{eq:IC}, each agent's report of $\xi$ affects his payoff only through the mechanism.\footnote{See, for example, analogous characterization results by \citet{jehielmoldovanu01}, \citet{cheetal13a}, and \citet{dworczaketal21}.}  Thus, no incentive-compatible mechanism can use reports of $\xi$ to screen agents; instead, the regulator must form beliefs about $\xi$ based only on agents' revealed consumption preferences.  Moreover, an adaptation of \citepos{myerson81} lemma implies that implementable allocation functions\footnote{Following the mechanism design literature, an allocation function $q:\Theta\times\R_+\to[0,A]$ is said to be {\em implementable} if there exists a payment function $t:\Theta\times\R_+\to\R$ such that the mechanism $(q,t)$ satisfies \eqref{eq:IC}.} must be nondecreasing in $\theta$ (cf.~\citealp{rochet87}):
\begin{lemma}\label{lem:IC}
Define
\[\calQ\coloneq\left\{q:\Theta\to[0,A]\text{ is nondecreasing}\right\}.\]
If an allocation function $q:\Theta\times\R_+\to[0,A]$ is implementable, then there exists $\hat q\in\calQ$ such that $q(\theta,\xi)=\hat q(\theta)$ for $F$-almost every $\theta\in\Theta$ and every $\xi\in\R_+$.
\end{lemma}

\Cref{lem:IC} is proven formally in \Cref{app:proofs} and shows that the formulation---and thus the solution---of the regulator's problem does not depend on whether agents observe the amount of externality they generate.  If agents do not observe $\xi$, the mechanism can condition only on information about $\theta$.  \Cref{lem:IC} shows that allowing agents to observe and report $\xi$ does not expand the set of allocations available to the regulator.  Henceforth, $\hat q\in\calQ$ is said to be \emph{implementable} if there exists an implementable $q$ such that $q(\theta,\xi)=\hat q(\theta)$ for $F$-almost every $\theta\in\Theta$ and every $\xi\in\R_+$.  Given any incentive-compatible $(q,t)$, the envelope theorem \citep{milgromsegal02} then implies that $t$ is likewise essentially independent of $\xi$.  With slight abuse of notation, $q$ and $t$ are henceforth written as functions of only $\theta$; throughout, allocation functions that agree $F$-almost everywhere are identified.

Given any incentive-compatible mechanism $(q,t)$, the regulator evaluates its performance with respect to the worst-case total surplus.  If the joint distribution $H$ of $(\theta,\xi)$ were known, then the total surplus would be\footnote{As monetary transfers are pure redistributions, total surplus depends on an incentive-compatible mechanism $(q,t)$ only through its allocation function $q$.  \Cref{sec:transfers} considers an extension where the regulator has revenue concerns.}
\[\int_\Theta \left[u(q(\theta),\theta)+\left[\E_H[\xi\mid\theta]-c\right]q(\theta)\right]\,\dd F(\theta).\]
However, the regulator does not observe the joint distribution $H$; she knows its marginal distribution $F$ and the unconditional mean $\E_H[\xi]=\mu>0$.  This reflects real-world settings in which policymakers know the average marginal benefit but, at most, have qualitative information about how marginal benefits depend on consumption preferences.  Below, the conditional mean of $\xi$ is denoted by $m(\theta)\coloneq\E_H[\xi\mid\theta]$, and the relationship between marginal benefits and consumption preferences is referred to as \emph{selection}.\footnote{As a special case, the formulation allows $\xi$ to be an unknown deterministic function of $\theta$: $\xi=m(\theta)$.}

To evaluate the worst-case total surplus, the regulator considers three informational benchmarks:
\begin{enumerate}[label=\emph{(\roman*)}]
\item \textbf{Unknown selection.}  The regulator allows all conditional mean functions $m$ in the set
\[\calM_0\coloneq\left\{m:\Theta\to\R_+\text{ such that $m$ is measurable and }\E_F[m(\theta)]=\mu\right\}.\]
This benchmark captures settings where the regulator has no reliable information about how marginal benefits depend on consumption preferences.

\item \textbf{Positive selection.}  The regulator allows all conditional mean functions $m$ in the set
\[\calM_+\coloneq\left\{m\in\calM_0:\E_F[\left[m(\theta)-\mu\right]\phi(\theta)]\geq 0\text{ for every bounded nondecreasing }\phi:\Theta\to\R\right\}.\]
If $d=1$, $\calM_+$ contains every nondecreasing $m\in\calM_0$.  More generally, marginal benefits tend to be larger among agents with stronger consumption preferences.  

\item \textbf{Negative selection.}  The regulator allows all conditional mean functions $m$ in the set
\[\calM_-\coloneq\left\{m\in\calM_0:\E_F[\left[m(\theta)-\mu\right]\phi(\theta)]\leq 0\text{ for every bounded nondecreasing }\phi:\Theta\to\R\right\}.\]
If $d=1$, $\calM_-$ contains every nonincreasing $m\in\calM_0$.  More generally, marginal benefits tend to be smaller among agents with stronger consumption preferences.  
\end{enumerate}

The last two benchmarks have a natural stochastic-order interpretation.  Given $m\in\calM_0$, define the \emph{externality-weighted distribution of types} $F^m$ for every measurable $B\subseteq\Theta$ by
\[F^m(B)\coloneq \frac{1}{\mu}\int_B m(\theta)\,\dd F(\theta).\]
Thus, $F^m(B)$ is the share of aggregate marginal externality generated by types in $B$.  Accordingly, $m\in\calM_+$ if and only if $F^m$ dominates $F$ in the usual multivariate stochastic order \citep[Chapter 6.B]{shakedshanthikumar07}, while $m\in\calM_-$ if and only if $F$ dominates $F^m$ in the same stochastic order.  Equivalently, every upper set accounts for at least its population share of aggregate marginal externality under positive selection and at most its population share under negative selection.

For each $*\in\{0,+,-\}$, the regulator's problem is therefore
\begin{equation}\label{eq:P}
\adjustlimits\sup_{q\text{ is implementable}}\inf_{m\in\calM_*}\int_\Theta\left[u(q(\theta),\theta)+\left[m(\theta)-c\right]q(\theta)\right]\,\dd F(\theta).\tag{P${}_*$}
\end{equation}
Problem~\eqref{eq:P} can be interpreted as a zero-sum game between the regulator and Nature.  The regulator first chooses an implementable allocation function $q$, after which Nature takes the infimum of total surplus over $m\in\calM_*$, subject to the constraints of the respective benchmark.

\subsection{Discussion of Assumptions}

To conclude this section, I discuss the key modeling choices underlying the framework and clarify their role in the subsequent analysis.

\paragraph{Information about the externality.}
A key assumption is that the regulator knows the average marginal externality $\mu=\E_H[\xi]$.  Beyond this mean, her information about how marginal benefits depend on consumption preferences is limited to the restrictions defining each informational benchmark.  This reflects policy settings in which average effects can be estimated from aggregate data even when individual-level marginal benefits are difficult to observe.  For example, epidemiological evidence may provide credible estimates of the average social benefit of vaccination, even when the relationship between that benefit and willingness to pay is much harder to estimate.

\paragraph{Information about demand.}
The model also assumes that the regulator knows the distribution of types, $F$.  This reflects the fact that consumption preferences can be estimated from observed consumption choices, whereas individual marginal social benefits or costs are typically harder to estimate.  Moreover, \Cref{sec:results} shows that the optimal regulation under each informational benchmark can be determined by solving at most a one-dimensional optimization problem.

\paragraph{Demand and supply specifications.}
The assumptions of continuous demand, constant marginal cost, and constant marginal externality are made for ease of exposition.  As \Cref{sec:extensions} shows, the subsequent analysis and main results extend to unit demand (e.g., vaccines).  In addition, while I follow standard textbook analyses by assuming that agent utility is additively separable and linear in the aggregate externality, neither restriction is essential.  More generally, under suitable monotonicity and curvature conditions, the form of the optimal regulation under each informational benchmark---and hence the comparison between price and quantity regulation---remains the same when agent utility takes the form $u(q,\theta,E)$, where $u$ is strictly increasing in $E$ but may be nonlinear and nonseparable in $(q,E)$.  Intuitively, nonseparability changes the shadow value of the aggregate externality---and hence the level of the optimal quantity floor or uniform subsidy---but not its form.

\section{Main Results}
\label{sec:results}

I now solve the regulator's problem under the three informational benchmarks---unknown, positive, and negative selection---and show how information about the dependence between consumption preferences and marginal externality shapes the optimal allocation.  

For ease of notation, let $D(p,\theta)$ denote the demand function of an agent of type $\theta$ when the per-unit price of the good is $p$.  The laissez-faire allocation is thus given by $q^{\LF}(\theta) = D(c,\theta)$.

\subsection{Unknown Selection}

I begin with the benchmark in which the regulator has no information about how the marginal externality depends on consumption preferences beyond its average level.

\begin{theorem}\label{thm:unknown}
Under the unknown-selection benchmark, there exists a quantity floor $\und q\in[0,A]$ such that the regulator's unique optimal allocation function is
\[q_0^*(\theta)= \max\left\{q^{\LF}(\theta),\und q\right\}.\]
\end{theorem}

\Cref{thm:unknown} shows that the optimal regulation has two important features that allow it to be interpreted as a form of quantity regulation.  First, it imposes a quantity floor: no agent can consume less than $\und q$ units of the good.  Second, it imposes no intervention above that floor: agents for whom the quantity floor does not bind simply consume their laissez-faire allocation.  

The first feature is intuitive: when selection is unknown, Nature can concentrate marginal benefit among agents assigned the smallest quantities.  A quantity floor raises the essential infimum of consumption and hence the worst-case external benefit.  At the same time, as \Cref{lem:IC} shows, incentive compatibility implies that any implementable allocation must be nondecreasing in type.  By imposing a quantity floor, the regulator can increase the worst-case external benefit while respecting incentive compatibility.

The second feature is perhaps more surprising: once the quantity floor is fixed, any intervention above it leaves the worst-case external benefit unchanged but distorts consumption away from laissez faire.  As such, laissez-faire consumption is optimal for every agent for whom the quantity floor does not bind.

The proof of \Cref{thm:unknown} formalizes these intuitions:  

\begin{proof}
To begin, observe that, for any allocation function $q\in\calQ$, Nature's problem satisfies
\[\inf_{m\in\calM_0} \E_F[m(\theta)q(\theta)]\geq \E_F[m(\theta)]\cdot\essinf_{\theta\in\Theta} q(\theta)=\mu\essinf_{\theta\in\Theta}q(\theta).\]
To show that this lower bound is tight, let $\und\theta\coloneq(\und\theta_1,\ldots,\und\theta_d)$ denote the lower corner of $\Theta$.  For each $n\in\NN$, define the set
\[B_n\coloneq\und\theta+\frac1n\(\Theta-\und\theta\).\]
Now consider the sequence $\{m_n\}_{n=1}^\infty\subset\calM_0$ defined by\footnote{For \Cref{thm:unknown} alone, a more natural choice might be $\hat m_n(\theta)=\mu\bone_{\theta\in B_n}/F(B_n)\in\calM_0$.  However, the present construction ensures that $m_n\in\calM_-$, allowing the same sequence to be used in the proof of \Cref{thm:sign}\emph{\ref{it:sign_neg}} below.}
\[m_n(\theta)\coloneq \mu n^d\frac{f\!\(\und\theta+n\(\theta-\und\theta\)\)}{f(\theta)}\cdot \bone_{\theta\in B_n}.\]
It is straightforward to verify that $\E_F[m_n(\theta)]=\mu$, so $m_n\in\calM_0$.  Moreover, by a change of variables,
\[\E_F[m_n(\theta)q(\theta)]=\mu\E_F\!\left[q\!\(\und\theta+\frac1n\(\theta-\und\theta\)\)\right].\]
Because $q$ is nondecreasing and $f$ is positive on $\Theta$,
\[q\!\(\und\theta+\frac1n\(\theta-\und\theta\)\)\to\essinf_{\hat\theta\in\Theta}q(\hat\theta),\qquad\text{for $F$-almost every $\theta\in\Theta$}.\]
Since $q$ is bounded, the bounded convergence theorem therefore implies
\begin{equation}\label{eq:BR}
\lim_{n\to\infty}\E_F[m_n(\theta)q(\theta)]=\mu\essinf_{\theta\in\Theta}q(\theta)\implies \inf_{m\in\calM_0}\E_F[m(\theta)q(\theta)]=\mu\essinf_{\theta\in\Theta}q(\theta).
\end{equation}
This shows that Nature can infimize external benefit by concentrating marginal benefit on agents assigned quantities arbitrarily close to the essential infimum of $q$.

Importantly, the worst-case external benefit in \cref{eq:BR} depends on the regulator's choice of $q$ only through its essential infimum, $\essinf_{\theta\in\Theta}q(\theta)$.  By \Cref{lem:IC}, replacing implementability with the requirement that $q\in\calQ$ yields a relaxation of the regulator's problem.  This relaxation can be rewritten as
\begin{align}
&\sup_{q\in\calQ}\left[\int_\Theta \left[u(q(\theta),\theta)-cq(\theta)\right]\,\dd F(\theta) + \mu\essinf_{\theta\in\Theta}q(\theta)\right]\tag*{}\\
&=\sup_{\und q\in[0,A]}\left[\sup_{q\in\calQ}\left\{\int_\Theta \left[u(q(\theta),\theta)-cq(\theta)\right]\,\dd F(\theta):\essinf_{\theta\in\Theta}q(\theta) \geq \und q\right\} + \mu \und q\right].\label{eq:reg}
\end{align}
For any choice of $\und q\in[0,A]$, the solution to the inner constrained problem in \cref{eq:reg} is given by the pointwise maximizer $q^{\LF}(\theta)$ unless the constraint is binding, yielding the desired solution form
\[q_0^*(\theta)=\max\left\{q^{\LF}(\theta),\und q\right\}.\]
Strict increasing differences imply that $q^{\LF}$ is nondecreasing, so $q_0^*\in\calQ$.  The objective of the outer problem is continuous in $\und q$, so an optimal quantity floor exists.  The solution to the relaxed problem is unique because $u$ is strictly concave in quantity, $\calQ$ is convex, and the essential infimum is concave in $q$.

It remains to verify that the solution to the relaxation is implementable.  Consider the pricing schedule
\[T_0^*(q)\coloneq\begin{dcases}
cq, &\text{if }q\geq\und q,\\
+\infty, &\text{if }q<\und q.
\end{dcases}\]
Under this pricing schedule, a type-$\theta$ agent chooses
\[\argmax_{q\in[\und q,A]}\left[u(q,\theta)-cq\right]=\max\left\{q^{\LF}(\theta),\und q\right\}=q_0^*(\theta).\]
Thus, $q_0^*$ is implementable.  Because $q_0^*$ solves the relaxed problem and is implementable, it also solves the regulator's problem (\hyperref[eq:P]{P${}_0$}).  Uniqueness of the solution follows from uniqueness in the relaxed problem.
\end{proof}

The proof of \Cref{thm:unknown} yields a simple first-order condition for an interior optimal quantity floor.  Specifically, the solution to the outer problem in \cref{eq:reg} satisfies
\begin{equation}\label{eq:FOC}
\E_F\!\left[\left[c-u_q(\und q,\theta)\right]\bone_{q^{\LF}(\theta)<\und q}\right]=\mu.
\end{equation}
In words, the optimal quantity floor equates the marginal loss in allocative efficiency from raising it---due to agents whose laissez-faire consumption falls below $\und q$---to the average marginal external benefit.

The key technical insight exploited in the proof is that the worst-case external benefit depends on the regulator's allocation function only through a one-dimensional statistic: its essential infimum.  This reduces the regulator's problem to a one-dimensional outer optimization over the quantity floor.

An implication of \Cref{thm:unknown} is that the distribution of types affects the level of the optimal quantity floor but not the form of the optimal regulation.  Computing the optimal floor requires evaluating \cref{eq:FOC}, and hence requires information about $F$.  Nevertheless, $F$ affects the choice of the floor only through the expected marginal loss in allocative efficiency from raising each candidate floor.

\subsection{Positive vs.~Negative Selection}

I next compare the benchmarks in which the regulator knows whether selection is positive or negative.

\begin{theorem}\label{thm:sign}\hfill
\begin{enumerate}[label=\emph{(\alph*)}]
\item\label{it:sign_pos} Under the positive-selection benchmark, the regulator's unique optimal allocation function can be implemented by a per-unit subsidy equal to $\mu$:
\[q_+^*(\theta)=D(c-\mu,\theta).\]

\item\label{it:sign_neg} Under the negative-selection benchmark, there exists a quantity floor $\und q\in[0,A]$ such that the regulator's unique optimal allocation function is
\[q_-^*(\theta)= q_0^*(\theta)=\max\left\{q^{\LF}(\theta),\und q\right\}.\]
\end{enumerate}
\end{theorem}

\Cref{thm:sign} shows that the direction of selection determines whether price or quantity regulation is optimal.  Under positive selection, the optimal regulation is a uniform subsidy equal to the average marginal benefit.  Under negative selection, the optimal regulation consists of a quantity floor.\footnote{When types are one-dimensional, \Cref{thm:sign} continues to hold if $\calM_+$ and $\calM_-$ are alternatively defined by requiring $m$ to be nondecreasing and nonincreasing, respectively.}

The negative-selection benchmark parallels the unknown-selection benchmark in \Cref{thm:unknown}.  In particular, the regulator must guard against the same worst-case possibility---namely, that marginal benefit is concentrated among agents assigned low quantities.  As such, the optimal allocation coincides with that in \Cref{thm:unknown} and again takes the form of a quantity floor above which there is no intervention.

The positive-selection benchmark is qualitatively different.  Nature cannot reduce external benefit by concentrating marginal benefit among agents assigned low quantities.  As the proof of \Cref{thm:sign} below shows, the constant conditional mean $m(\theta)\equiv\mu$ is instead a best response to every $q\in\calQ$.  The regulator therefore evaluates external benefit as if marginal benefit were homogeneous and equal to $\mu$, and the optimal policy reduces to a uniform subsidy.

Notably, \emph{no} information about the distribution of types is required to implement the optimal regulation under the positive-selection benchmark.  Thus, the optimal regulation in this case is also robust to demand uncertainty.

The proof of \Cref{thm:sign} formalizes these arguments:

\begin{proof}
To prove the result for the positive-selection benchmark, observe that the worst-case external benefit can be bounded more tightly than in \Cref{thm:unknown}.  Given the regulator's choice of allocation function $q\in\calQ$,
\[\inf_{m\in\calM_+} \E_F[m(\theta)q(\theta)] \geq \E_F[m(\theta)]\cdot\E_F[q(\theta)]=\mu\E_F[q(\theta)].\]
The inequality above follows directly from the definition of $\calM_+$ because $q$ is bounded and nondecreasing.  To show that equality is in fact attained, consider Nature's best response: $m_+^*(\theta)\equiv\mu$.  Given this strategy,
\[\E_F[m_+^*(\theta)q(\theta)] = \mu \E_F[q(\theta)] = \inf_{m\in\calM_+}\E_F[m(\theta)q(\theta)].\]
Thus, $m_+^*$ is a best response to every $q\in\calQ$ and corresponds to no selection.

As in the proof of \Cref{thm:unknown}, the worst-case external benefit under the positive-selection benchmark depends on the regulator's choice of $q$ only through a one-dimensional statistic: in this case, its average, $\E_F[q(\theta)]$.  By \Cref{lem:IC}, replacing implementability with the requirement that $q\in\calQ$ yields the relaxed problem
\[\sup_{q\in\calQ}\left[\int_\Theta \left[u(q(\theta),\theta)-cq(\theta)\right]\,\dd F(\theta) + \mu\E_F[q(\theta)]\right]=\sup_{q\in\calQ}\int_\Theta \left[u(q(\theta),\theta)-\(c-\mu\)q(\theta)\right]\,\dd F(\theta).\]
The solution to the relaxed problem is therefore given by the pointwise maximizer $q_+^*(\theta) = D(c-\mu,\theta)$.  Strict increasing differences imply that $q_+^*$ is nondecreasing, so $q_+^*\in\calQ$.  Moreover, $q_+^*$ is implemented by the pricing schedule $T_+^*(q)=(c-\mu)q$.  Thus, $q_+^*$ solves the regulator's problem (\hyperref[eq:P]{P${}_+$}).  As in the proof of \Cref{thm:unknown}, uniqueness of the solution follows from the strict concavity of $u$.

Finally, consider the sequence $\{m_n\}_{n=1}^\infty$ constructed in the proof of \Cref{thm:unknown}.  For every bounded nondecreasing function $\phi:\Theta\to\R$, the same change of variables used in that proof yields
\[\E_F\!\left[\left[m_n(\theta)-\mu\right]\phi(\theta)\right]=\mu\E_F\!\left[\phi\!\(\und\theta+\frac1n\(\theta-\und\theta\)\)-\phi(\theta)\right]\leq0,\]
where the inequality follows because
\[\und\theta+\frac1n\(\theta-\und\theta\)\leq\theta.\]
Hence, $m_n\in\calM_-$ for every $n\in\NN$.  Since $\calM_-\subseteq\calM_0$, \cref{eq:BR} therefore implies
\[\inf_{m\in\calM_-}\E_F[m(\theta)q(\theta)]=\mu\essinf_{\theta\in\Theta}q(\theta).\]
Thus, the relaxed problem under the negative-selection benchmark is identical to that under the unknown-selection benchmark.  Its unique solution is thus $q_-^*=q_0^*$, which is implementable by \Cref{thm:unknown}.
\end{proof}

\subsection{Discussion of Main Results}

The proofs of \Cref{thm:unknown,thm:sign} clarify the connection between the robust formulation and its Bayesian counterpart.  Under positive selection, $m_+^*(\theta)\equiv\mu$ attains the infimum for every $q\in\calQ$; under unknown and negative selection, the same sequence $\{m_n\}_{n=1}^\infty$ approaches the infimum for every $q\in\calQ$.  Because neither $m_+^*$ nor $\{m_n\}_{n=1}^\infty$ depends on the regulator's allocation function, the robust objective coincides with the Bayesian objective evaluated at $m_+^*$, or with the pointwise limit of the Bayesian objectives evaluated along $\{m_n\}_{n=1}^\infty$.\footnote{\Cref{sec:support} considers an extension in which there need not exist a single conditional mean function that is least favorable against every implementable allocation, but the regulator's problem can nonetheless be reduced to a one-dimensional optimization problem.}

The form of these worst-case conditional mean functions makes the regulator's problem tractable by reducing external benefit to a one-dimensional statistic of $q$.  Under positive selection, external benefit reduces to $\mu\E_F[q(\theta)]$, yielding a uniform subsidy equal to $\mu$.  Under unknown and negative selection, limiting external benefit reduces to $\mu\essinf_{\theta\in\Theta}q(\theta)$, yielding a quantity floor above which consumption remains undistorted.  Information about selection therefore determines both the statistic of the allocation function that matters in the worst case and the form of the optimal regulation.

This reduction does not generally extend to the Bayesian problem with a known conditional mean function $m$.  Even though the allocation is one-dimensional, the Bayesian problem need not admit a simple characterization when types are multidimensional \citep{rochetstole03}.  In particular, monotonicity is necessary but not sufficient for implementability, so the pointwise solution to the relaxed problem need not be implementable \citep{rochet87}.  The Bayesian-optimal policy may therefore involve a more complicated nonlinear subsidy schedule \citep{kaplowshavell02}.  Moreover, even among uniform subsidies, an interior optimal subsidy is generally a demand-response-weighted average of marginal benefits rather than their unconditional mean $\mu$ \citep{diamond73}.  The simple policy forms in \Cref{thm:unknown,thm:sign} therefore follow from evaluating regulation against the worst-case conditional mean functions above, rather than from a restriction on the regulator's instruments.

While I have focused on the case of positive externalities, the same reasoning applies to negative externalities.  Specifically, suppose that a type-$\theta$ agent who consumes a quantity $q(\theta,\xi)$ of the good, generates a negative externality of $\xi$ per unit of consumption, and makes a payment of $t(\theta,\xi)$ obtains utility
\[u(q(\theta,\xi),\theta) - t(\theta,\xi) - E(q),\quad\text{where }E(q)\coloneq\iint_{(\hat\theta,\hat\xi)\in\Theta\times\R_+}\hat\xi q(\hat\theta,\hat\xi)\,\dd H(\hat\theta,\hat\xi).\]

Relative to the case of positive externalities, quantity ceilings---rather than quantity floors---become optimal when there are negative externalities.  The counterpart of \Cref{thm:unknown} is:
\begin{retheorem}\label{rethm:unknown}
Under the unknown-selection benchmark for negative externalities, there exists a quantity ceiling $\BAR q\in[0,A]$ such that the regulator's unique optimal allocation function is
\[q_0^*(\theta)= \min\left\{q^{\LF}(\theta),\BAR q\right\}.\]
\end{retheorem}
\begin{proof}
See \Cref{app:proofs}.
\end{proof}

Moreover, the roles of positive and negative selection are reversed:

\begin{retheorem}\label{rethm:sign}
\hfill
\begin{enumerate}[label=\emph{(\alph*)}]
\item Under the positive-selection benchmark for negative externalities, there exists a quantity ceiling $\BAR q\in[0,A]$ such that the regulator's unique optimal allocation function is
\[q_+^*(\theta)= q_0^*(\theta)=\min\left\{q^{\LF}(\theta),\BAR q\right\}.\]
\item Under the negative-selection benchmark for negative externalities, the regulator's unique optimal allocation function can be implemented by a per-unit tax equal to $\mu$:
\[q_-^*(\theta)=D(c+\mu,\theta).\]
\end{enumerate}
\end{retheorem}
\begin{proof}
See \Cref{app:proofs}.
\end{proof}

With negative externalities, Nature seeks to maximize rather than minimize $\E_F[m(\theta)q(\theta)]$.  Under unknown or positive selection, Nature can concentrate marginal harm among agents assigned the largest quantities, so worst-case harm depends only on the essential supremum of $q$.  The regulator therefore imposes a quantity ceiling, leaving laissez-faire consumption unchanged below the ceiling.  Under negative selection, by contrast, the constant conditional mean $m(\theta)\equiv\mu$ maximizes external harm, and a uniform tax equal to $\mu$ is robustly optimal.  Thus, absent information supporting negative selection, the robust criterion favors quantity regulation---quotas, limits, or bans---over price regulation.

\section{Extensions}
\label{sec:extensions}

Returning to positive externalities, this section extends the main results of \Cref{sec:results} to three settings: unit demand, revenue concerns, and bounded support for the marginal externality.  Proofs are deferred to \Cref{app:proofs}.

\subsection{Unit Demand}\label{sec:unit}

I first extend the analysis to the case of unit demand.  Each agent demands at most one unit of the good, and his willingness to pay is denoted by $v(\theta)$, where $v:\Theta\to\R_+$ is continuous and increasing in the product order.  Interpreting $q(\theta)\in[0,1]$ as the probability that a type-$\theta$ agent consumes the good, his direct utility under an allocation $(q,t)$ is $v(\theta)q(\theta)-t(\theta)$.  Thus, utility is linear rather than strictly concave in $q$, as assumed in \Cref{sec:framework}.  The laissez-faire allocation is denoted by $q^{\LF}(\theta)=\bone_{v(\theta)>c}$.

Consumption generates a marginal external benefit $\xi\in\R_+$ that varies across agents.  Vaccination provides a natural application: $v(\theta)$ represents an individual's willingness to pay for vaccination, while $\xi$ represents the marginal external benefit generated by vaccinating him, which may depend on his exposure to infection or propensity to transmit it.

\begin{proposition}\label{prop:unit}
Consider the unit-demand model.
\begin{enumerate}[label=\emph{(\alph*)}]
\item Under the unknown-selection benchmark and the negative-selection benchmark, an optimal allocation function is
\[q_0^*(\theta)=q_-^*(\theta)=\begin{dcases}
1, &\text{if }\mu\geq\E_F\!\left[\left[c-v(\theta)\right]_+\right],\\
q^{\LF}(\theta), &\text{if }\mu<\E_F\!\left[\left[c-v(\theta)\right]_+\right].
\end{dcases}\]

\item Under the positive-selection benchmark, an optimal allocation function can be implemented by a uniform subsidy equal to $\mu$:
\[q_+^*(\theta)=\bone_{v(\theta)>c-\mu}.\]
\end{enumerate}
\end{proposition}

\Cref{prop:unit} shows that it suffices for the regulator to consider only three candidates for the optimal policy: a mandate ($q^*\equiv 1$), no intervention ($q^*=q^{\LF}$), and a subsidy equal to the average external benefit $\mu$.  A subsidy is optimal when selection is positive.  By contrast, if selection is negative or unknown, then either a mandate or no intervention is optimal.

The regulator's choice between a mandate and no intervention is determined by a comparison between the average allocative deadweight loss imposed by a mandate and the average external benefit it guarantees.  Relative to laissez faire, mandating consumption by a type-$\theta$ agent imposes an allocative deadweight loss of $\left[c-v(\theta)\right]_+$; thus, $\E_F\!\left[\left[c-v(\theta)\right]_+\right]$ represents the average loss in allocative efficiency generated by a mandate.  When the average external benefit $\mu$ weakly exceeds this loss, a mandate is optimal.  Otherwise, the allocative loss from any intervention weakly exceeds the additional external benefit it guarantees in the worst case, so the regulator optimally refrains from intervening.  This is the unit-demand analog of the tradeoff characterized in \cref{eq:FOC}.

For vaccines, the direction of selection describes whether marginal external benefits tend to be larger or smaller among individuals with greater willingness to pay.  In the case of COVID-19, individuals with greater willingness to pay may also have taken greater precautions, potentially generating negative selection.  By contrast, positive selection may be more plausible for HPV or mpox vaccines, as individuals who anticipate greater exposure to infection may also have greater willingness to pay.

\subsection{Revenue Concerns}\label{sec:transfers}

Whereas the framework of \Cref{sec:framework} treats monetary transfers as pure redistributions, I now extend the analysis to the case of a regulator with revenue concerns.  As in \Cref{sec:framework}, suppliers are competitive and earn zero profits, so payments net of production costs accrue to the regulator.  

Since regulation may involve mandates and taxes, I do not require that agents be able to freely opt out.  Instead, I impose a lower bound on aggregate direct utility, which can be interpreted as a political constraint on the burden that regulation can impose on agents as a group.  Specifically, fix $\la>0$ and $\und U\in\R$.  The regulator places weight $\la$ on policy revenue---that is, revenue net of production costs---and requires aggregate direct utility to be at least $\und U$.\footnote{If the constraint instead imposed a lower bound on worst-case aggregate agent utility including the external benefit, the same characterization would obtain after reparameterizing $\la$.}  For each $*\in\{0,+,-\}$, the regulator's problem is thus
\begin{align}\label{eq:revenue}
\adjustlimits\sup_{(q,t)\text{ satisfies }\eqref{eq:IC}}\inf_{m\in\calM_*}&\left\{\E_F[m(\theta)q(\theta)]+\la\int_\Theta\left[t(\theta)-cq(\theta)\right]\,\dd F(\theta)\right\},\tag{R${}_*$}\\
\text{s.t. } &\int_\Theta\left[u(q(\theta),\theta)-t(\theta)\right]\,\dd F(\theta)\geq \und U.\notag
\end{align}

The following proposition shows that this extension changes the intensity, but not the form, of the optimal regulation.

\begin{proposition}\label{prop:transfers}
Consider the model with revenue concerns.
\begin{enumerate}[label=\emph{(\alph*)}]
\item Under the unknown-selection benchmark and the negative-selection benchmark, there exists a quantity floor $\und q_\la\in[0,A]$ such that the regulator's unique optimal allocation function is
\[q_{0,\,\la}^*(\theta)=q_{-,\,\la}^*(\theta)=\max\left\{q^{\LF}(\theta),\und q_\la\right\}.\]

\item Under the positive-selection benchmark, the regulator's unique optimal allocation function can be implemented by a per-unit subsidy equal to $\mu/\la$:
\[q_{+,\,\la}^*(\theta)=D\!\(c-\frac{\mu}{\la},\theta\).\]
\end{enumerate}
\end{proposition}

The key observation in the proof of \Cref{prop:transfers} is that adding the same constant to every payment preserves incentive compatibility and leaves the induced allocation unchanged.  Since $\la>0$, the regulator therefore increases payments until the aggregate utility constraint binds.  It follows that
\[\int_\Theta\left[t(\theta)-cq(\theta)\right]\,\dd F(\theta)=\int_\Theta\left[u(q(\theta),\theta)-cq(\theta)\right]\,\dd F(\theta)-\und U.\]
The regulator's allocation problem \eqref{eq:revenue} is equivalent (up to rescaling and shifting) to
\[\sup_{q\text{ is implementable}}\left\{\int_\Theta\left[u(q(\theta),\theta)-cq(\theta)\right]\,\dd F(\theta)+\frac{1}{\la}\inf_{m\in\calM_*}\E_F[m(\theta)q(\theta)]\right\}.\]
In turn, this problem is identical to (\hyperref[eq:P]{P${}_*$}) after replacing $\mu$ with $\mu/\la$, which yields \Cref{prop:transfers}.

\Cref{prop:transfers} generalizes the results of \Cref{thm:unknown,thm:sign} by showing how the optimal policy changes with the weight $\la$ on policy revenue.  When $\la=1$, the regulator's objective is equal to the worst-case total surplus (up to rescaling and shifting), so the optimal allocations coincide with those in \Cref{thm:unknown,thm:sign}.  As $\la$ decreases, the regulator places relatively more weight on external benefit: the optimal quantity floor weakly increases, while the subsidy under positive selection, $\mu/\la$, becomes larger.  The opposite occurs as $\la$ increases.  Thus, under the formulation in \eqref{eq:revenue}, revenue concerns affect the intensity of optimal regulation but not its form.

\subsection{Bounded Support}\label{sec:support}

Finally, I extend the analysis to the case in which marginal external benefits are known to have bounded support \emph{a priori}.  Specifically, suppose that the regulator knows that $\xi\in[\und\xi,\BAR\xi]\subset\R_+$ almost surely, where $0\leq\und\xi\leq\mu\leq\BAR\xi$.  For each $*\in\{0,+,-\}$, the corresponding ambiguity set becomes
\[\calM_*(\und\xi,\BAR\xi)\coloneq\left\{m\in\calM_*:\und\xi\leq m(\theta)\leq\BAR\xi\text{ for $F$-almost every $\theta\in\Theta$}\right\}.\]

\begin{proposition}\label{prop:support}
Consider the model with bounded support.  Under the unknown-selection benchmark, there exists a threshold quantity $\und q\in[0,A]$ such that the regulator's unique optimal allocation function is
\[q_0^*(\theta;\und\xi,\BAR\xi)=\max\left\{D(c-\und\xi,\theta),\min\left\{\und q,D(c-\BAR\xi,\theta)\right\}\right\}.\]
\end{proposition}

Although the optimal allocation in \Cref{prop:support} is more complicated than the quantity-floor allocation in \Cref{thm:unknown}, it can be implemented by a simple two-tier subsidy:
\[T_0^*(q;\und\xi,\BAR\xi)\coloneq\big(c-\und\xi\big)\,q-\big(\BAR\xi-\und\xi\big)\,\min\{q,\und q\}.\]
Equivalently, the regulator subsidizes the first $\und q$ units at the rate $\BAR\xi$ and every additional unit at the rate $\und\xi$.  Thus, the lower support bound determines the subsidy that applies to all consumption, while the difference between the upper and lower support bounds determines an additional subsidy on initial units.  Agents whose demand at the lower price $c-\BAR\xi$ is below $\und q$ consume $D(c-\BAR\xi,\theta)$; agents with intermediate demand consume $\und q$; and agents whose demand at the higher price $c-\und\xi$ exceeds $\und q$ consume $D(c-\und\xi,\theta)$.

For a sufficiently large upper support bound, the two-tier subsidy implements a quantity floor.  To see this, suppose that $\BAR\xi\geq c$.  Then $D(c-\BAR\xi,\theta)=A$ for every $\theta$, so the allocation in \Cref{prop:support} reduces to
\[q_0^*(\theta;\und\xi,\BAR\xi)=\max\left\{D(c-\und\xi,\theta),\und q\right\}.\]
Moreover, when $\und\xi=0$, $D(c-\und\xi,\theta)=q^{\LF}(\theta)$, and the regulator's choice of $\und q$ is identical to the choice of the quantity floor in \cref{eq:reg}.  Hence, $q_0^*(\theta;0,\BAR\xi)$ coincides with the optimal allocation in \Cref{thm:unknown} for every $\BAR\xi\geq c$.  More generally, when $\und\xi>0$, the optimal allocation combines a uniform subsidy equal to $\und\xi$ with a quantity floor.

This result also clarifies the economic mechanism underlying \Cref{thm:unknown}.  With a finite upper support bound, Nature cannot concentrate the marginal external benefit above $\und\xi$ on an arbitrarily small set of agents.  The regulator responds by subsidizing initial units at the largest possible marginal external benefit, $\BAR\xi$, thereby protecting against the concentration of marginal external benefits among agents assigned the lowest quantities.  When $\BAR\xi<c$, this subsidy may be insufficient to induce agents with the weakest demand to consume $\und q$.  Once $\BAR\xi\geq c$, however, the subsidized marginal price is nonpositive, so every agent is induced to consume at least $\und q$ and the two-tier subsidy imposes a quantity floor. 

The proof of \Cref{prop:support} differs in a subtle but important respect from that of \Cref{thm:unknown}.  In the proof of \Cref{thm:unknown}, Nature chooses the same infimizing sequence of conditional mean functions against \emph{any} implementable allocation function.  By contrast, with bounded support, Nature's best response generally depends on the regulator's allocation: it assigns a marginal external benefit of $\BAR\xi$ to agents assigned the lowest quantities and $\und\xi$ to agents assigned the highest quantities, with an intermediate value at the cutoff when needed to satisfy the mean restriction.  With multidimensional types, different implementable allocations may rank incomparable types differently, so there is generally no single conditional mean function that yields the worst case against every implementable allocation.  The proof of \Cref{prop:support} thus solves Nature's problem for each fixed allocation before optimizing over the regulator's allocation.  

The support restriction has no effect under positive selection.  The constant conditional mean $m(\theta)\equiv\mu$ remains feasible and is a best response to every implementable allocation function, so the uniform subsidy characterized in \Cref{thm:sign} remains optimal for every $0\leq\und\xi\leq\mu\leq\BAR\xi$.

Under negative selection, the argument underlying \Cref{prop:support} applies when $d=1$, but need not extend to multidimensional types.  With one-dimensional types, Nature's assignment of the remaining marginal external benefit to agents assigned the smallest quantities is consistent with negative selection.  In higher dimensions, the same assignment need not satisfy the stochastic-order restriction defining $\calM_-$.  Nevertheless, I show in the Supplemental Appendix that as $\BAR\xi\to+\infty$, the allocation combining a uniform subsidy equal to $\und\xi$ with the corresponding quantity floor is asymptotically optimal.

\section{Conclusion}

Much of the ``prices versus quantities'' debate starts from \citepos{weitzman74} classic benchmark in which the regulator compares relatively simple instruments under relatively rich knowledge of primitives.  A subsequent literature emphasizes that, once a broader set of instruments is allowed, price-based policies can dominate quantity controls \citep{kaplowshavell02}.

This paper departs from that benchmark on both dimensions.  The regulator is not restricted to simple instruments: through a mechanism-design formulation of the instrument choice problem, she may choose any nonlinear pricing schedule.  At the same time, the regulator knows the distribution of consumption preferences and the average marginal externality, but has limited information about how the two are related.  She therefore evaluates policies by their worst-case welfare performance given this information.

The first main result provides a new explanation for the prevalence of quantity regulation based on heterogeneity and robustness.  With a positive externality and unknown selection, Nature can approach the worst case by concentrating marginal benefit among agents assigned the smallest quantities.  Quantity regulation then delivers a welfare guarantee by directly controlling externality-relevant behavior: the unique optimal allocation can be implemented by a quantity floor.  With a negative externality, Nature can instead approach the worst case by concentrating marginal harm among agents assigned the largest quantities, and the unique optimal allocation can be implemented by a quantity ceiling.  With unit demand, these controls can take the form of mandates and bans.

The second main result characterizes when price regulation is justified and why it should be simple.  With a positive externality, positive selection rules out concentration of marginal benefit among agents assigned the smallest quantities, and the optimal allocation can be implemented by a uniform subsidy equal to the average marginal externality.  Under negative selection, the worst-case pattern under unknown selection remains feasible, so a quantity floor remains optimal. The roles of positive and negative selection are reversed for negative externalities.  Thus, instrument choice turns not on whether the regulator can implement rich pricing schedules---she can---but on her information about the direction of selection.

A natural question is whether the selection patterns that matter for robustness are economically meaningful.  In many environments, consumption preferences and marginal externalities are shaped by common underlying factors.  For example, individuals who are most eager to get vaccinated may already engage in extensive precautionary behavior, reducing the marginal social benefit of their vaccination; conversely, individuals with the lowest willingness to pay may be those whose behavior generates the largest marginal social benefit.  Empirical evidence also suggests that externalities are often driven by tail behavior.  \citet{griffithetal19} note that alcohol-related externalities are highly nonlinear in consumption, while 8\% of alcohol drinkers in the United States account for 51\% of total drinks consumed \citep{esseretal20}.  Similarly, \citet{knittelsandler18} find strong skew in tailpipe emissions and in the marginal social costs of accidents and congestion.  These patterns are consistent with the concentration of externalities that drives the worst cases characterized in this paper.

The results also extend in several directions.  Revenue concerns affect the intensity---but not the form---of optimal regulation.  With known bounds on marginal external benefits, the optimal policy under unknown selection is a two-tier subsidy that implements a quantity floor when the upper support bound becomes sufficiently large.  Thus, while the exact policy depends on the regulator's objective and information, the basic rationale for quantity regulation remains intact.

Finally, the same robustness argument applies beyond heterogeneous consumption externalities.  The Supplemental Appendix shows that robustness favors a quantity control under regulatory uncertainty \emph{\`a la} \citet{weitzman74}, even when the regulator may choose any nonlinear payment schedule.  In environments with costly and socially wasteful screening, it favors random allocation without screening, such as lotteries rather than ordeals \citep{nicholszeckhauser82}, and may also help explain procedures that avoid costly litigation.  These applications show how robustness under limited information can justify simple policies even when richer instruments are available.  Exploring other applications and examining the role of information acquisition and commitment in mitigating robustness concerns are promising directions for future research.

\clearpage

\putbib[master_bibliography]

\clearpage

\appendix

\section{Omitted Proofs}
\label{app:proofs}

\subsection{Proof of \texorpdfstring{\Cref{lem:IC}}{Lemma 1}}

Suppose that $q$ is implementable by some payment function $t$.  For any $(\theta,\xi),(\theta',\xi')\in\Theta\times\R_+$, \eqref{eq:IC} requires that
\[\begin{dcases}
u(q(\theta,\xi),\theta)-t(\theta,\xi) &\geq u(q(\theta',\xi'),\theta)-t(\theta',\xi'),\\
u(q(\theta',\xi'),\theta')-t(\theta',\xi') &\geq u(q(\theta,\xi),\theta')-t(\theta,\xi).\\
\end{dcases}\]
Adding both inequalities yields
\[u(q(\theta,\xi),\theta)-u(q(\theta',\xi'),\theta)\geq u(q(\theta,\xi),\theta')-u(q(\theta',\xi'),\theta').\]
Let $\theta\geq \theta'$ and $\theta\neq\theta'$.  Since $u$ satisfies strict increasing differences in $(q,\theta)$, the above inequality implies that $q(\theta,\xi)\geq q(\theta',\xi')$ for any $\xi,\xi'\in\R_+$.  Fix $\xi_0\in\R_+$ and define $\hat q\coloneq q(\cdot,\xi_0)$; thus, $\hat q\in\calQ$.

It remains to show that $q(\theta,\xi)=\hat q(\theta)$ for $F$-almost every $\theta\in\Theta$ and every $\xi\in\R_+$.  Because $\hat q$ is coordinatewise nondecreasing, it is continuous at Lebesgue-almost every point in the interior of $\Theta$---and hence at $F$-almost every $\theta\in\Theta$, since $F$ admits a density.  Fix any such $\theta$ and let $e=(1,\ldots,1)$.  For every $\xi\in\R_+$ and every sufficiently small $\e>0$, the earlier inequality implies that
\[\hat q(\theta-\e e) \leq q(\theta,\xi)\leq \hat q(\theta+\e e).\]
Letting $\e\downarrow 0$ yields $q(\theta,\xi)=\hat q(\theta)$ for $F$-almost every $\theta\in\Theta$ and every $\xi\in\R_+$, as claimed.

\subsection{Proof of \texorpdfstring{\Cref{rethm:unknown}'}{Theorem 1'}}

In the case of negative externalities, observe that, for any allocation function $q\in\calQ$, Nature's problem satisfies
\[\sup_{m\in\calM_0}\E_F[m(\theta)q(\theta)]\leq\E_F[m(\theta)]\cdot\esssup_{\theta\in\Theta}q(\theta)=\mu\esssup_{\theta\in\Theta}q(\theta).\]
To show that this upper bound is tight, let
$\BAR\theta\coloneq(\BAR\theta_1,\ldots,\BAR\theta_d)$ denote the upper corner of $\Theta$.  For each $n\in\NN$, define the set
\[B_n\coloneq \BAR\theta-\frac1n\(\BAR\theta-\Theta\).\]
Now consider the sequence $\{m_n\}_{n=1}^\infty\subset\calM_0$ defined by
\[m_n(\theta)\coloneq\mu n^d\frac{f\!\(\BAR\theta-n\(\BAR\theta-\theta\)\)}{f(\theta)}\cdot\bone_{\theta\in B_n}.\]
It is straightforward to verify that $\E_F[m_n(\theta)]=\mu$, so $m_n\in\calM_0$.  Moreover, by a change of variables,
\[\E_F[m_n(\theta)q(\theta)]=\mu\E_F\!\left[q\!\(\BAR\theta-\frac1n\(\BAR\theta-\theta\)\)\right].\]
Because $q$ is nondecreasing and $f$ is positive on $\Theta$,
\[q\!\(\BAR\theta-\frac1n\(\BAR\theta-\theta\)\)\to\esssup_{\hat\theta\in\Theta}q(\hat\theta),\qquad\text{for $F$-almost every $\theta\in\Theta$}.\]
Since $q$ is bounded, the bounded convergence theorem therefore implies
\begin{equation}\label{eq:BR_neg}
\lim_{n\to\infty}\E_F[m_n(\theta)q(\theta)]=\mu\esssup_{\theta\in\Theta}q(\theta)\implies\sup_{m\in\calM_0}\E_F[m(\theta)q(\theta)]=\mu\esssup_{\theta\in\Theta}q(\theta).
\end{equation}
This shows that Nature can maximize external harm in the limit by concentrating marginal harm on agents assigned quantities arbitrarily close to the essential supremum of $q$.

Importantly, the worst-case external harm in \cref{eq:BR_neg} depends on the regulator's choice of $q$ only through its essential supremum, $\esssup_{\theta\in\Theta}q(\theta)$.  By \Cref{lem:IC}, replacing implementability with the requirement that $q\in\calQ$ yields a relaxation of the regulator's problem.  This relaxation can be rewritten as
\begin{align}
&\sup_{q\in\calQ}\left[\int_\Theta\left[u(q(\theta),\theta)-cq(\theta)\right]\,\dd F(\theta)-\mu\esssup_{\theta\in\Theta}q(\theta)\right]\notag\\
&=\sup_{\BAR q\in[0,A]}\left[\sup_{q\in\calQ}\left\{\int_\Theta\left[u(q(\theta),\theta)-cq(\theta)\right]\,\dd F(\theta):\esssup_{\theta\in\Theta}q(\theta)\leq\BAR q\right\}-\mu\BAR q\right].
\label{eq:reg_neg}
\end{align}
For any choice of $\BAR q\in[0,A]$, the solution to the inner constrained problem in \cref{eq:reg_neg} is given by the pointwise maximizer $q^{\LF}(\theta)$ unless the constraint is binding, yielding the desired solution form
\[q_0^*(\theta)=\min\left\{q^{\LF}(\theta),\BAR q\right\}.\]
Strict increasing differences imply that $q^{\LF}$ is nondecreasing, so $q_0^*\in\calQ$.  The objective of the outer problem is continuous in $\BAR q$, so an optimal quantity ceiling exists.  The solution to the relaxed problem is unique because $u$ is strictly concave in quantity, $\calQ$ is convex, and the essential supremum is convex in $q$.

It remains to verify that the solution to the relaxation is implementable.  Consider the pricing schedule
\[T_0^*(q)\coloneq\begin{dcases}
cq, &\text{if }q\leq\BAR q,\\
+\infty, &\text{if }q>\BAR q.
\end{dcases}\]
Under this pricing schedule, a type-$\theta$ agent chooses
\[\argmax_{q\in[0,\BAR q]}\left[u(q,\theta)-cq\right]=\min\left\{q^{\LF}(\theta),\BAR q\right\}=q_0^*(\theta).\]
Thus, $q_0^*$ is implementable.  Because $q_0^*$ solves the relaxed problem and is implementable, it also solves the regulator's problem under the unknown-selection benchmark.  Uniqueness of the solution follows from uniqueness in the relaxed problem.

\subsection{Proof of \texorpdfstring{\Cref{rethm:sign}'}{Theorem 2'}}

To prove the result for the negative-selection benchmark, observe that the worst-case external harm can be bounded more tightly than in \Cref{rethm:unknown}'.  Given the regulator's choice of allocation function $q\in\calQ$, the definition of $\calM_-$ implies that
\[\sup_{m\in\calM_-}\E_F[m(\theta)q(\theta)]\leq\E_F[m(\theta)]\cdot\E_F[q(\theta)]=\mu\E_F[q(\theta)].\]
To show that equality is in fact attained, consider Nature's best response: $m_-^*(\theta)\equiv \mu$.  Given this strategy,
\[\E_F[m_-^*(\theta)q(\theta)]=\mu\E_F[q(\theta)]=\sup_{m\in\calM_-}\E_F[m(\theta)q(\theta)].\]
Thus, $m_-^*$ is a best response to every $q\in\calQ$ and corresponds to no selection.

As in the proof of \Cref{rethm:unknown}', the worst-case external harm under the negative-selection benchmark depends on the regulator's choice of $q$ only through a one-dimensional statistic: in this case, its average, $\E_F[q(\theta)]$.  By \Cref{lem:IC}, replacing implementability with the requirement that $q\in\calQ$ yields the relaxed problem
\[\sup_{q\in\calQ}\left[\int_\Theta\left[u(q(\theta),\theta)-cq(\theta)\right]\,\dd F(\theta)-\mu\E_F[q(\theta)]\right]=\sup_{q\in\calQ}\int_\Theta\left[u(q(\theta),\theta)-\(c+\mu\)q(\theta)\right]\,\dd F(\theta).\]
The solution to the relaxed problem is thus given by the pointwise maximizer $q_-^*(\theta)=D(c+\mu,\theta)$.  Strict increasing differences imply that $q_-^*$ is nondecreasing, so $q_-^*\in\calQ$.  Moreover, $q_-^*$ is implemented by the pricing schedule $T_-^*(q)=(c+\mu)q$.  Thus, $q_-^*$ solves the regulator's problem under the negative-selection benchmark.  As in the proof of \Cref{thm:sign}, uniqueness follows from the strict concavity of $u$.

Finally, consider the sequence $\{m_n\}_{n=1}^\infty$ constructed in the proof of \Cref{rethm:unknown}'.  For every bounded nondecreasing function $\phi:\Theta\to\R$, the same change of variables used in that proof yields
\[\E_F\!\left[\left[m_n(\theta)-\mu\right]\phi(\theta)\right]=\mu\E_F\!\left[\phi\!\(\BAR\theta-\frac1n\(\BAR\theta-\theta\)\)-\phi(\theta)\right]\geq0,\]
where the inequality follows because
\[\BAR\theta-\frac1n\(\BAR\theta-\theta\)\geq\theta.\]
Hence, $m_n\in\calM_+$ for every $n\in\NN$.  Since $\calM_+\subseteq\calM_0$, \cref{eq:BR_neg} therefore implies
\[\sup_{m\in\calM_+}\E_F[m(\theta)q(\theta)]=\mu\esssup_{\theta\in\Theta}q(\theta).\]
Thus, the relaxed problem under the positive-selection benchmark is identical to that under the unknown-selection benchmark.  Its unique solution is therefore $q_+^*=q_0^*$, which is implementable by \Cref{rethm:unknown}'.

\subsection{Proof of \texorpdfstring{\Cref{prop:unit}}{Proposition 1}}

Let
\[\calQ_{[0,1]}\coloneq\left\{q:\Theta\to[0,1]\text{ is nondecreasing}\right\}.\]
Since \Cref{lem:IC} continues to apply in the unit-demand model, replacing implementability with the requirement that $q\in\calQ_{[0,1]}$ yields a relaxation of the regulator's problem.

\begin{enumerate}[label={(\alph*)}]
\item The calculations in \cref{eq:BR} and in the proof of \Cref{thm:sign} imply that, for every $q\in\calQ_{[0,1]}$,
\[\inf_{m\in\calM_0}\E_F[m(\theta)q(\theta)]=\inf_{m\in\calM_-}\E_F[m(\theta)q(\theta)]=\mu\essinf_{\theta\in\Theta}q(\theta).\]
Thus, under either unknown or negative selection, the relaxed problem can be rewritten as
\begin{align}
&\sup_{q\in\calQ_{[0,1]}}\left[\int_\Theta\left[v(\theta)-c\right]q(\theta)\,\dd F(\theta)+\mu\essinf_{\theta\in\Theta}q(\theta)\right]\notag\\
&=\sup_{\und q\in[0,1]}\left[\sup_{q\in\calQ_{[0,1]}}\left\{\int_\Theta\left[v(\theta)-c\right]q(\theta)\,\dd F(\theta):\essinf_{\theta\in\Theta}q(\theta)\geq \und q\right\}+\mu \und q\right].\label{eq:reg_unit}
\end{align}
For any choice of $\und q\in[0,1]$, a solution to the inner constrained problem in \cref{eq:reg_unit} is
\[q^*(\theta;\und q)=\max\left\{q^{\LF}(\theta),\und q\right\}.\]
It is easy to verify that this is implementable for any $\und q\in[0,1]$.  Substituting this into the outer problem yields
\[\sup_{\und q\in[0,1]}\left\{\E_F\!\left[\left[v(\theta)-c\right]_+\right]+\left[\mu-\E_F\!\left[\left[c-v(\theta)\right]_+\right]\right]\und q\right\}.\]
Observe that the objective is affine in $\und q$.  If $\mu\geq\E_F\!\left[\left[c-v(\theta)\right]_+\right]$, then $\und q=1$ is optimal, yielding a mandate.  If $\mu<\E_F\!\left[\left[c-v(\theta)\right]_+\right]$, then $\und q=0$ is optimal, yielding the laissez-faire allocation.

\item Under positive selection, the proof of \Cref{thm:sign} implies that, for every $q\in\calQ_{[0,1]}$,
\[\inf_{m\in\calM_+}\E_F[m(\theta)q(\theta)]=\mu\E_F[q(\theta)].\]
The relaxed problem is therefore
\[\sup_{q\in\calQ_{[0,1]}}\int_\Theta\left[v(\theta)-c+\mu\right]q(\theta)\,\dd F(\theta).\]
It is easy to check that the pointwise maximizer $q_+^*(\theta)=\bone_{v(\theta)>c-\mu}$ is implemented by the pricing schedule $T_+^*(q)=\(c-\mu\)q$ and hence solves the regulator's problem under the positive-selection benchmark.
\end{enumerate}

\subsection{Proof of \texorpdfstring{\Cref{prop:transfers}}{Proposition 2}}

Fix an implementable allocation function $q$ and a corresponding payment function $t$.  Adding the same constant to every payment preserves incentive compatibility and leaves the induced allocation unchanged.  Since $\la>0$, the regulator optimally increases payments until the aggregate utility constraint binds:
\[\int_\Theta t(\theta)\,\dd F(\theta)=\int_\Theta u(q(\theta),\theta)\,\dd F(\theta)-\und U.\]
Substituting this expression into the regulator's objective, dropping the constant $-\la\und U$, and dividing by $\la$ reduces the allocation problem to
\[\sup_{q\text{ is implementable}}\left\{\int_\Theta\left[u(q(\theta),\theta)-cq(\theta)\right]\,\dd F(\theta)+\frac1\la\inf_{m\in\calM_*}\E_F\!\left[m(\theta)q(\theta)\right]\right\}.\]
For each $*\in\{0,+,-\}$, multiplying an admissible conditional mean function by $1/\la$ preserves the relevant selection restriction and changes its mean from $\mu$ to $\mu/\la$.  The allocation problem is therefore identical to (\hyperref[eq:P]{P${}_*$}) after replacing $\mu$ with $\mu/\la$.  The desired result follows from \Cref{thm:unknown,thm:sign}.

\subsection{Proof of \texorpdfstring{\Cref{prop:support}}{Proposition 3}}

If $\und\xi=\BAR\xi$, then $\und\xi=\mu=\BAR\xi$ and the unique optimal allocation function is $q_0^*(\theta)=D(c-\mu,\theta)$; the result follows immediately by taking $\und q=A$.  Below, I assume that $\und\xi<\BAR\xi$.

To begin, consider the worst-case external benefit that Nature can choose for any fixed allocation function $q\in\calQ$.  For every $m\in\calM_0(\und\xi,\BAR\xi)$ and every $x\in[0,A]$,
\begin{align*}
\E_F[m(\theta)q(\theta)]
&=\und\xi\E_F[q(\theta)] + \big(\mu-\und\xi\big)\,x + \E_F\!\left[\big[m(\theta)-\und\xi\big]\left[q(\theta)-x\right]\right]\\
&\geq\und\xi\E_F[q(\theta)] + \big(\mu-\und\xi\big)\,x - \big(\BAR\xi-\und\xi\big)\E_F\!\left[\left[x-q(\theta)\right]_+\right].
\end{align*}
The inequality follows because $0\leq m(\theta)-\und\xi\leq\BAR\xi-\und\xi$.  Hence,
\[\inf_{m\in\calM_0(\und\xi,\BAR\xi)}\E_F[m(\theta)q(\theta)]\geq\und\xi\E_F[q(\theta)] + \max_{x\in[0,A]}\left\{\big(\mu-\und\xi\big)\,x - \big(\BAR\xi-\und\xi\big)\E_F\!\left[\left[x-q(\theta)\right]_+\right]\right\}.\]

I now show that this lower bound is tight for any fixed allocation function $q\in\calQ$.  To this end, choose $x\in[0,A]$ such that
\[\Pr_F[q(\theta)<x]\leq\frac{\mu-\und\xi}{\BAR\xi-\und\xi}\leq\Pr_F[q(\theta)\leq x].\]
Moreover, consider
\[m(\theta)=\und\xi\cdot\bone_{q(\theta)>x}+m_0\cdot\bone_{q(\theta)=x}+\BAR\xi\cdot\bone_{q(\theta)<x},\]
where $m_0$ is chosen so that $\E_F[m(\theta)]=\mu$.  It is easy to see that $m\in\calM_0(\und\xi,\BAR\xi)$ and attains the preceding lower bound: by construction,
\begin{align*}
\E_F[m(\theta)\left[q(\theta)-x\right]] 
&= \und\xi\E_F\!\left[\left[q(\theta)-x\right]\bone_{q(\theta)>x}\right] + \BAR\xi\E_F\!\left[\left[q(\theta)-x\right]\bone_{q(\theta)<x}\right]\\
&=\und\xi\E_F\!\left[q(\theta)-x\right] - \big(\BAR\xi-\und\xi\big)\E_F\!\left[\left[x-q(\theta)\right]\bone_{q(\theta)<x}\right].
\end{align*}
Thus,
\[\E_F[m(\theta)q(\theta)] = \und\xi\E_F[q(\theta)] + \big(\mu-\und\xi\big)\,x - \big(\BAR\xi-\und\xi\big)\E_F[\left[x-q(\theta)\right]_+].\]
Combining this equality with the preceding lower bound shows that the chosen $x$ maximizes the expression over $[0,A]$ and that $m$ attains the infimum.  It follows that
\begin{equation}\label{eq:BR_support}
\inf_{m\in\calM_0(\und\xi,\BAR\xi)}
\E_F[m(\theta)q(\theta)]=\und\xi\E_F[q(\theta)]+\max_{x\in[0,A]}\left\{\big(\mu-\und\xi\big)x-\big(\BAR\xi-\und\xi\big)\E_F\!\left[\left[x-q(\theta)\right]_+\right]\right\}.
\end{equation}

By \Cref{lem:IC}, replacing implementability with the requirement that $q\in\calQ$ yields a relaxation of the regulator's problem.  Using \cref{eq:BR_support}, this relaxation can be rewritten as
\[\sup_{x\in[0,A]}\left[\big(\mu-\und\xi\big)\,x+\sup_{q\in\calQ}\int_\Theta\left[u(q(\theta),\theta)-\big(c-\und\xi\big)\,q(\theta)-\big(\BAR\xi-\und\xi\big)\left[x-q(\theta)\right]_+\right]\,\dd F(\theta)\right].\]
For any choice of $x\in[0,A]$, the unique pointwise maximizer of the inner problem is
\[\max\left\{D(c-\und\xi,\theta),\min\left\{x,D(c-\BAR\xi,\theta)\right\}\right\}.\]
Strict increasing differences imply that this allocation is nondecreasing, so it belongs to $\calQ$.  The objective of the outer problem is continuous in $x$, so an optimal threshold quantity $\und q\in[0,A]$ exists.  The resulting allocation is 
\[q_0^*(\theta;\und\xi,\BAR\xi)=\max\left\{D(c-\und\xi,\theta),\min\left\{\und q,D(c-\BAR\xi,\theta)\right\}\right\}.\]

It remains to verify that the solution to the relaxation is implementable.  To this end, consider the pricing schedule
\[T_0^*(q;\und\xi,\BAR\xi)\coloneq\big(c-\und\xi\big)\,q-\big(\BAR\xi-\und\xi\big)\min\{q,\und q\}.\]
Under this pricing schedule, a type-$\theta$ agent chooses
\[\argmax_{q\in[0,A]}
\left[u(q,\theta)-T_0^*(q;\und\xi,\BAR\xi)\right]=\max\left\{D(c-\und\xi,\theta),\min\left\{\und q,D(c-\BAR\xi,\theta)\right\}\right\}.\]
Thus, $q_0^*(\cdot;\und\xi,\BAR\xi)$ is implementable and, because it solves the relaxed problem, also solves the regulator's problem under the unknown-selection benchmark.  Finally, strict concavity of $u$, convexity of $\calQ$, and concavity of the worst-case external benefit in $q$ imply that the optimal allocation is unique.

\end{bibunit}

\clearpage

\begin{bibunit}[econ-econometrica]
\setbibunitprefix{supp}

\appendix
\renewcommand{\thepage}{SA-\arabic{page}}
\setcounter{page}{1}
\setcounter{footnote}{0}

\numberwithin{equation}{section}
\numberwithin{thmno}{section}
\numberwithin{propno}{section}

\renewcommand*{\theHsection}{supp.\arabic{section}}
\renewcommand*{\theHsubsection}{\theHsection.\arabic{subsection}}
\renewcommand*{\theHsubsubsection}{\theHsubsection.\arabic{subsubsection}}
\renewcommand*{\theHequation}{\theHsection.\arabic{equation}}
\renewcommand*{\theHthmno}{\theHsection.\arabic{thmno}}
\renewcommand*{\theHpropno}{\theHsection.\arabic{propno}}

\begin{center}
{\Large\bfseries Supplemental Appendix}
\end{center}

\vspace{1em}

\section{Bounded Support under Known Selection}
\label{app:support_selection}

Whereas the analysis of bounded support in \Cref{sec:support} of the main paper focuses on unknown selection, this section examines the cases of positive and negative selection.

\subsection{Positive Selection}

Under positive selection, the support restriction has no effect.  For every implementable allocation function $q$, the definition of $\calM_+$ implies that
\[
\inf_{m\in\calM_+(\und\xi,\BAR\xi)}
\E_F[m(\theta)q(\theta)]
\geq
\mu\E_F[q(\theta)].
\]
The constant conditional mean $m(\theta)\equiv\mu$ belongs to $\calM_+(\und\xi,\BAR\xi)$ and attains this lower bound.  Hence, the uniform subsidy characterized in \Cref{thm:sign} remains optimal.

\subsection{Negative Selection}

When $d=1$, the characterization in \Cref{prop:support} also remains valid under negative selection.  The cases in which $\mu$ equals a support bound are immediate, so suppose that $\und\xi<\mu<\BAR\xi$ and define
\[m(\theta)\coloneq\begin{dcases}
\und\xi, &\text{if }\theta>F^{-1}\(\frac{\mu-\und\xi}{\BAR\xi-\und\xi}\),\\
\BAR\xi, &\text{if }\theta\leq F^{-1}\(\frac{\mu-\und\xi}{\BAR\xi-\und\xi}\).\end{dcases}\]
By construction, $\E_F[m(\theta)]=\mu$ and $m$ is nonincreasing, so $m\in\calM_-(\und\xi,\BAR\xi)$.  For every implementable allocation function $q$ and every $\hat m\in\calM_0(\und\xi,\BAR\xi)$,
\[\left[\hat m(\theta)-m(\theta)\right]\left[q(\theta)-q\!\(F^{-1}\(\frac{\mu-\und\xi}{\BAR\xi-\und\xi}\)\)\right]\geq0,\qquad\text{for $F$-almost every $\theta\in\Theta$}.\]
Because implementability implies that $q$ is nondecreasing, taking expectations and using $\E_F[\hat m(\theta)-m(\theta)]=0$, 
\[\E_F[\hat m(\theta)q(\theta)] \geq \E_F[m(\theta)q(\theta)],\qquad\text{for every implementable $q$}.\]
Thus, $m$ minimizes expected external benefit over $\calM_0(\und\xi,\BAR\xi)$ and is itself feasible under negative selection.  It follows that
\[\inf_{m\in\calM_-(\und\xi,\BAR\xi)}\E_F[m(\theta)q(\theta)]=\inf_{m\in\calM_0(\und\xi,\BAR\xi)}\E_F[m(\theta)q(\theta)].\]
Consequently, \Cref{prop:support} applies without change.

When $d>1$, the preceding argument need not apply, but the combination of a uniform subsidy equal to $\und\xi$ with a quantity floor becomes asymptotically optimal as $\BAR\xi\to+\infty$.  To determine the relevant quantity floor, suppose that the upper support restriction is removed while the lower support bound is retained.  Applying the argument in the proof of \Cref{thm:sign} to $m(\theta)-\und\xi$---which is nonnegative and has mean $\mu-\und\xi$---yields the regulator's problem
\[\sup_{q\text{ is implementable}}\left\{\int_\Theta\left[u(q(\theta),\theta)-\big(c-\und\xi\big)\,q(\theta)\right]\,\dd F(\theta)+\big(\mu-\und\xi\big)\,\essinf_{\theta\in\Theta}q(\theta)\right\}.\]
Replacing implementability with the requirement that $q\in\calQ$ yields a relaxation.  As in the proof of \Cref{thm:unknown}, there exists some $\und q\in[0,A]$ such that its unique solution takes the form
\[\max\left\{D(c-\und\xi,\theta),\und q\right\}.\]
Because this allocation can be implemented by combining a uniform subsidy equal to $\und\xi$ with a quantity floor, the relaxation is tight.

\clearpage
\begin{proposition}\label{prop:support_limit}
Fix $0\leq\und\xi<\mu$.  Under negative selection, the allocation combining a uniform subsidy equal to $\und\xi$ with the corresponding quantity floor is asymptotically optimal as $\BAR\xi\to+\infty$: the difference between optimal worst-case welfare and the worst-case welfare it delivers converges to zero.
\end{proposition}

\begin{proof}
Denote this allocation by $q^*$.  For every $q\in\calQ$, define the worst-case welfare delivered by $q$ as
\[W_{\BAR\xi}(q)\coloneq\int_\Theta\left[u(q(\theta),\theta)-cq(\theta)\right]\,\dd F(\theta)+\inf_{m\in\calM_-(\und\xi,\BAR\xi)}\E_F[m(\theta)q(\theta)].\]
Let $W_\infty(q)$ denote the objective in the limiting problem as $\BAR\xi\to+\infty$.  Now, for every $m\in\calM_-(\und\xi,\BAR\xi)$,
\[\E_F[m(\theta)q(\theta)]\geq\und\xi\E_F[q(\theta)]+\big(\mu-\und\xi\big)\,\essinf_{\theta\in\Theta}q(\theta)=\inf_{\hat m\in\calM_0(\und\xi,+\infty)}\E_F[\hat m(\theta)q(\theta)].\]
Consequently, $W_{\BAR\xi}(q)\geq W_\infty(q)$ for every $q\in\calQ$, and hence
\[\sup_{q\text{ is implementable}}W_{\BAR\xi}(q)\geq W_{\BAR\xi}(q^*)\geq W_\infty(q^*)=\sup_{q\in\calQ}W_\infty(q).\]
To establish the reverse inequality in the limit, for each $n\geq2$, let
\[L_n\coloneq\left[\und\theta,\und\theta+\frac1n\,\big(\BAR\theta-\und\theta\big)\right]\AND H_n\coloneq\left(\und\theta+\frac1n\,\big(\BAR\theta-\und\theta\big),\BAR\theta\right].\]
Define
\[m_n(\theta)\coloneq\mu+\big(\mu-\und\xi\big)\left[\frac{F(H_n)}{F(L_n)}\bone_{\theta\in L_n}-\bone_{\theta\in H_n}\right].\]
Each $m_n$ is bounded, satisfies $m_n\geq\und\xi$, and has mean $\mu$.  Moreover, for every bounded nondecreasing function $\phi$,
\[\E_F[\left[m_n(\theta)-\mu\right]\phi(\theta)]=\big(\mu-\und\xi\big)\,F(H_n)\left[\E_F[\phi(\theta)\mid\theta\in L_n]-\E_F[\phi(\theta)\mid\theta\in H_n]\right]\leq0.\]
This holds because every type in $L_n$ is weakly smaller than every type in $H_n$.  As such, $m_n\in\calM_-$ and, because it is bounded, $m_n\in\calM_-(\und\xi,\BAR\xi)$ for every sufficiently large $\BAR\xi$.  Consequently, for every $q\in\calQ$,
\begin{align*}
\E_F[m_n(\theta)q(\theta)]
&=\mu\E_F[q(\theta)]+\big(\mu-\und\xi\big)\,F(H_n)\left[\E_F[q(\theta)\mid\theta\in L_n]-\E_F[q(\theta)\mid\theta\in H_n]\right]\\
&\to\und\xi\E_F[q(\theta)]+\big(\mu-\und\xi\big)\,\essinf_{\theta\in\Theta}q(\theta).
\end{align*}
The convergence follows because $F(H_n)\to1$, $L_n$ shrinks to the lower corner of $\Theta$, and $q$ is bounded and nondecreasing.  Together with the preceding lower bound, this implies that
\[W_{\BAR\xi}(q)\longrightarrow W_\infty(q),\qquad\text{for every fixed $q\in\calQ$}.\]

Finally, observe that $\BAR\xi\mapsto\sup_{q\text{ is implementable}}W_{\BAR\xi}(q)$ is nonincreasing.  Let $\BAR\xi_k\to+\infty$, and choose an implementable allocation $q_k$ whose worst-case welfare is within $1/k$ of the optimal value.  By Helly's selection theorem, a subsequence converges $F$-almost everywhere to some $q\in\calQ$.  For each fixed $n$, $m_n\in\calM_-(\und\xi,\BAR\xi_k)$ for all sufficiently large $k$.  Therefore,
\[\sup_{\hat q\text{ is implementable}}W_{\BAR\xi_k}(\hat q)\leq\int_\Theta\left[u(q_k(\theta),\theta)-cq_k(\theta)\right]\,\dd F(\theta)+\E_F[m_n(\theta)q_k(\theta)]+\frac1k.\]
Taking $k\to\infty$ along this subsequence, applying dominated convergence, and then letting $n\to\infty$ gives
\[\adjustlimits\lim_{\BAR\xi\to+\infty}\sup_{\hat q\text{ is implementable}}W_{\BAR\xi}(\hat q)\leq W_\infty(q)\leq W_\infty(q^*).\]
Together with the preceding lower bound, this shows that the optimal value converges to $W_\infty(q^*)$.  Finally,
\[0\leq\sup_{q\text{ is implementable}}W_{\BAR\xi}(q)-W_{\BAR\xi}(q^*)\leq\sup_{q\text{ is implementable}}W_{\BAR\xi}(q)-W_\infty(q^*)\to0.\qedhere\]
\end{proof}

\section{Application to Regulatory Uncertainty}\label{sec:uncertainty}

While the analysis in the main paper focuses on heterogeneity across agents, I now apply the framework to uncertainty faced by the regulator---a classic problem studied by \cite{weitzman74} and a large subsequent literature in environmental economics, law and economics, and other fields.  In this setting, my results show that robustness with respect to uncertainty supports the use of quantity regulation.

Consider a regulator who wishes to regulate an externality-producing industry by reducing the amount of the externality it produces.  Let $q\in[0,A]$ denote the level of externality reduction, where $A>0$ is the maximum feasible level of reduction.  For simplicity, suppose that the marginal benefit of externality reduction is constant.  This assumption is a reasonable approximation when the industry has a small overall impact on the environment and is commonly imposed in textbook treatments of the problem (see, e.g., Chapter 21 of \citealp{viscusietal05}).  By contrast, the industry's marginal cost of externality reduction is increasing, reflecting rising abatement costs.

The regulator faces uncertainty over both the cost and benefit of externality reduction.  This uncertainty is captured by the random variables $\theta\in\Theta\subset\R$ and $\xi\in\R_+$, which represent the aggregate industry's cost and benefit shifters.  The cost function $C:[0,A]\times\Theta\to\R_+$ is increasing and strictly convex in $q$ ($C_q>0$ and $C_{qq}>0$), and higher values of $\theta$ correspond to higher marginal costs ($C_{q\theta}>0$).  It is also assumed to be integrable: $\E_F[C(A,\theta)]<+\infty$.  Likewise, $\xi$ shifts the marginal benefit curve: the marginal benefit of reducing the externality by an additional unit is $\xi$.  

Importantly, the direction of selection between cost and benefit shifters may be unknown.  For example, a global demand shock can raise both the marginal benefits (because baseline emissions from other industries are higher) and the marginal costs of abatement (because reductions must be achieved at higher output levels).  Conversely, if scarce innovation effort must be allocated between abatement technologies (which lower $\theta$) and mitigation or capture technologies (which lower $\xi$), then there may be negative selection.  \cite{stavins96} discusses other examples of how negative selection may arise.

This setup departs from the previous literature along two dimensions: the set of feasible policies and the regulator's objective.  

First, the regulator chooses a potentially nonlinear payment schedule $T:[0,A]\to\R\cup\{+\infty\}$, under which the industry pays $T(q)$ when it reduces $q$ units of the externality.  In particular, the regulator is not restricted to particular instruments (such as constant marginal prices or fixed quantity instruments), as in the models of \citeauthor{weitzman74} and \citeauthor{stavins96}---and indeed much of the subsequent literature, as critiqued by \cite{kaplowshavell02}.  

Second, the regulator does not know the joint distribution of $(\theta,\xi)$ and does not treat uncertainty over $\theta$ and $\xi$ symmetrically.  As \citeauthor{stavins96} notes, ``more often than not, it is benefit uncertainty that seems to be of substantially greater magnitude.''  For example, while policymakers often have access to information about marginal costs through abatement technologies and compliance costs, marginal benefits depend on more complex processes---such as environmental damages, health impacts, or ecological responses---that are difficult to measure and may vary substantially across states of the world.  Accordingly, the regulator is assumed to know only $\mu=\E[\xi]$ and the marginal distribution $F$ of $\theta$, and she evaluates each policy according to its worst-case ex-ante welfare performance.

The game proceeds as follows:
\begin{enumerate}
\item The regulator commits to a payment schedule $T:[0,A]\to\R\cup\{+\infty\}$.
\item The cost and benefit shifters $\theta$ and $\xi$ are realized and observed by the industry.
\item The industry chooses the level of externality reduction $q$ and pays $T(q)$.
\end{enumerate}
For any payment schedule $T$, fix a tie-breaking rule that does not depend on $\xi$, and denote the induced allocation by 
\[q_T(\theta)\in\argmin_{q\in[0,A]}\left[C(q,\theta)+T(q)\right].\]
Because $\xi$ does not enter the industry's objective, the induced allocation depends on $\theta$ but not on $\xi$.  Let $m(\theta)\coloneq\E[\xi\mid\theta]$.  The regulator's problem can therefore be written as
\[\adjustlimits\sup_T\inf_{m\in\calM_0}\E_F\!\left[m(\theta)q_T(\theta)-C(q_T(\theta),\theta)\right].\]

While this environment differs from that in the main paper, the same robustness argument applies.  Specifically, Nature can approach the worst case by concentrating the marginal benefit of externality reduction on states assigned the smallest quantities.  This yields the following result.

\begin{proposition}\label{prop:uncertainty}
There exists a unique quantity $\und q\in[0,A]$ such that the regulator's unique optimal allocation is $q^*(\theta)=\und q$.  It can be implemented by the payment schedule
\[T^*(q)=\begin{dcases}
0, &\text{if }q\geq\und q,\\
+\infty, &\text{if }q<\und q.
\end{dcases}\]
\end{proposition}

\begin{proof}
Replacing the set of allocation functions induced by nonlinear payment schedules with the set of all measurable functions $q:\Theta\to[0,A]$ yields a relaxation of the regulator's problem.  For any such allocation function,
\[\inf_{m\in\calM_0}\E_F[m(\theta)q(\theta)]\geq \mu\essinf_{\theta\in\Theta}q(\theta).\]
To show that it is tight, let $q_{\min}\coloneq\essinf q$ and define
\[B_n\coloneq\left\{\theta\in\Theta:q(\theta)<q_{\min}+\frac1n\right\}\AND m_n(\theta)\coloneq\frac{\mu}{F(B_n)}\bone_{\theta\in B_n}.\]
Then $m_n\in\calM_0$ and
\[\E_F[m_n(\theta)q(\theta)]\leq\mu\(q_{\min}+\frac1n\).\]
Taking $n\to\infty$ establishes the equality.  Moreover, since $C$ is increasing in its first argument, the following holds for any measurable allocation function $q$:
\[\mu\essinf q-\E_F[C(q(\theta),\theta)]\leq\mu q_{\min}-\E_F[C(q_{\min},\theta)].\]
Thus, the value of the relaxed problem is bounded above by
\[\max_{q_{\min}\in[0,A]}\left\{\mu q_{\min}-\E_F[C(q_{\min},\theta)]\right\}.\]
Strict convexity of $C$ implies that this problem has a unique solution $\und q$.

It remains to show that the solution to the relaxation is induced by a payment schedule.  Under $T^*$, the industry must choose $q\geq\und q$.  Since $C$ is strictly increasing in $q$, it uniquely chooses $q(\theta)=\und q$ for every $\theta$.  Thus, $T^*$ implements the solution to the relaxed problem and therefore solves the regulator's original problem.  Strict monotonicity of $C$ also implies that no nonconstant allocation can attain the same value, establishing uniqueness.
\end{proof}

Even though the regulator can, in principle, choose any nonlinear payment schedule, \Cref{prop:uncertainty} shows that the optimal policy takes the form of a quantity control.  Under the optimal payment schedule $T^*$, the industry chooses the minimum level of externality reduction $q=\und q$ in every realization of $(\theta,\xi)$.  Because externality reduction is costly, favorable cost realizations reduce the cost of meeting the requirement but do not induce additional abatement beyond $\und q$.  Thus, robustness collapses the optimal policy to a noncontingent quantity rule as considered by \citeauthor{weitzman74} and \citeauthor{stavins96}, despite the regulator's ability to employ a much richer class of instruments.

As in \Cref{thm:unknown}, the quantity floor $\und q$ in \Cref{prop:uncertainty} can be characterized by a simple first-order condition whenever it is interior:
\begin{equation}\label{eq:quantity}
\E_F[C_q(\und q,\theta)] = \mu.
\end{equation}
When the expected marginal cost always exceeds the expected marginal benefit (\ie, $\E_F[C_q(0,\theta)]>\mu$), the optimal policy is no intervention.  Conversely, when the expected marginal cost is always less than the expected marginal benefit (\ie, $\E_F[C_q(A,\theta)]<\mu$), the optimal policy requires the maximum feasible level of externality reduction. 

Despite the intuitive appeal of the first-order condition \eqref{eq:quantity}, it should be emphasized that it generally does \emph{not} arise in Bayesian models where the joint distribution is known.  As noted by \citet{robertsspence76} and \citet{kaplowshavell02}, for example, the pointwise solution to the relaxed Bayesian problem satisfies
$C_q(q(\theta),\theta)=m(\theta)$ whenever $q(\theta)$ is interior, rather than in expectation over $\theta$.

The optimal policy characterized in \Cref{prop:uncertainty} also admits a natural interpretation as a cap-and-trade policy.  For a fixed baseline level of emissions, the quantity floor $\und q$ on abatement levels equivalently imposes a cap on the total amount of externality produced by the industry, which can be implemented by allocating permits to firms in the industry.  While trading permits is not explicitly modeled, it helps allocate abatement efforts more efficiently across firms within the industry, which determines the industry's cost function $C(q,\theta)$.  This logic mirrors the operation of real-world cap-and-trade programs such as the U.S.~Acid Rain Program for sulfur dioxide, the European Union Emissions Trading System, California's cap-and-trade program, and the Regional Greenhouse Gas Initiative.

While \Cref{prop:uncertainty} focuses on unknown selection, the logic of \Cref{thm:sign} extends to this setting with the roles of positive and negative selection reversed.  Because every allocation induced by a payment schedule is nonincreasing in $\theta$, negative selection between cost and benefit shifters---so that states in which abatement is cheaper tend also to be those in which abatement is more valuable---makes the constant conditional mean $m(\theta)\equiv\mu$ a worst case.  A per-unit subsidy for externality reduction equal to $\mu$ is therefore optimal.  As noted previously, negative selection could arise if scarce innovation effort must be allocated between abatement and capture technologies.  However, when the direction of selection is uncertain (for example, advances in abatement might have positive spillovers into capture technologies and vice versa), then the analysis favors quantity regulation.

Beyond environmental regulations, \Cref{prop:uncertainty} also has implications for a broader set of policy settings in which regulators face uncertainty about costs and benefits.  The insights of \citet{weitzman74} have been applied to tariff policy in international trade (\eg, \citealp{dasguptastiglitz77} and \citealp{younganderson80}), to safety regulation (\eg, \citealp{shavell84}), and to innovation policy (\eg, \citealp{wright83}).  In these settings, as in the present framework, regulators choose between price and quantity instruments under limited information, and the robustness considerations developed here provide a complementary perspective on why quantity regulations are often observed in practice.

Finally, it is important to emphasize the limited normative scope of \Cref{prop:uncertainty}.  The analysis characterizes the policy a regulator would choose when information about costs and benefits is scarce and the objective is to maximize a worst-case welfare guarantee.

\section{Application to Costly Screening}\label{app:costly_screening}

The final section of the Supplemental Appendix applies the robustness argument to environments in which goods are allocated through costly and socially wasteful screening, such as waiting in line.  Costly screening has been used to study a wide range of economic settings, including bargaining, lobbying, and redistribution programs.  I show that limited information about how agents' willingness to wait is related to their waiting costs can rationalize random allocation without screening.

There is a quantity $Q\in(0,1)$ of a good to be allocated to a unit mass of risk-neutral agents.  Each agent demands at most one unit and has a type $(\theta,w)$, where $\theta\in\Theta\coloneq[\und\theta,\BAR\theta]\subset\R_+$ and $w\in\R_{++}$.  The agent's value for the good is $v=\theta w$, while $w$ is his cost per unit of waiting time.  Thus, given a probability $q\in[0,1]$ of receiving the good after waiting for time $t\in\R_+$, his payoff is $w\cdot \(\theta q-t\)$.

The regulator knows the distribution $F$ of $\theta$ and the average waiting cost $\mu=\E[w]$, but has no additional information about how waiting costs depend on $\theta$.  As in the main paper, $F$ admits a positive density on $\Theta$.  Let $m(\theta)\coloneq\E[w\mid\theta]$ and define\footnote{Allowing $m$ to vanish does not affect the worst-case value: for any such $m$ and $\e\in(0,1)$, $\(1-\e\)m+\e\mu$ is strictly positive, has mean $\mu$, and yields the same limiting expected surplus as $\e\downarrow0$.}
\[\calM_0\coloneq\left\{m:\Theta\to\R_{+}\text{ such that $m$ is measurable and }\E_F[m(\theta)]=\mu\right\}.\]

The regulator chooses an allocation function $q:\Theta\to[0,1]$ and a waiting-time function $t:\Theta\to\R_+$.  Because $w>0$ simply scales the agent's payoff, the analog of \Cref{lem:IC} implies that any incentive-compatible mechanism is essentially independent of $w$:
\begin{equation}\label{eq:IC_costly}
\theta q(\theta)-t(\theta)\geq\theta q(\hat\theta)-t(\hat\theta),\qquad\text{for every }\theta,\hat\theta\in\Theta.
\tag{IC}
\end{equation}
The regulator seeks to maximize worst-case aggregate surplus.  Thus, her problem is
\begin{equation}\label{eq:P_costly}
\sup_{(q,t)\text{ satisfies \eqref{eq:IC_costly}}}\left\{\inf_{m\in\calM_0}\E_F\!\left[m(\theta)\left[\theta q(\theta)-t(\theta)\right]\right]:\E_F[q(\theta)]\leq Q\right\}.\tag{P}
\end{equation}

Although this environment differs from that in the main paper, the same robustness logic applies.  The uncertain primitive enters agents' payoffs directly through waiting costs, so Nature's worst-case choice depends on the induced utility schedule rather than on the allocation function alone.

\begin{proposition}\label{prop:costly}
Under the unknown-selection benchmark, the mechanism $(q^*,t^*)=(Q,0)$ is optimal.  Moreover, if $\und\theta>0$, it is the unique optimal mechanism up to $F$-null sets.
\end{proposition}

\begin{proof}
Define $U(\theta)\coloneq\theta q(\theta)-t(\theta)$.  Standard arguments imply that $q$ is nondecreasing and
\[U(\theta)=U(\und\theta)+\int_{\und\theta}^{\theta}q(s)\,\dd s.\]
Thus, $U$ is nondecreasing.  Consequently, for every $m\in\calM_0$,
\[\E_F[m(\theta)U(\theta)]\geq\mu\essinf_{\theta\in\Theta}U(\theta).\]
To show that this bound is tight, consider the sequence $\{m_n\}_{n=1}^\infty\subset\calM_0$ defined by
\[m_n(\theta) \coloneq \begin{dcases}
0, &\text{if }F(\theta)>1/n,\\
n\mu, &\text{otherwise}.
\end{dcases}\]
Given the regulator's choice of allocation function $q$, for any $\e>0$, $\Pr_F[U(\theta)<\essinf_{\theta\in\Theta}U(\theta)+\e]>0$ by the definition of the essential infimum.  Consequently, there exists $N_\e\in\NN$ such that
\[\Pr_F\!\left[U(\theta)<\essinf_{\theta\in\Theta}U(\theta)+\e\right]\geq 1/N_\e.\]
Since $U$ is nondecreasing, this choice of $N_\e$ yields
\[\E_F[m_{N_\e}(\theta)U(\theta)] \leq \mu\left[\essinf_{\theta\in\Theta}U(\theta)+\e\right].\] 
In turn, this implies that 
\[\lim_{\e\to 0}\E_F[m_{N_\e}(\theta)U(\theta)]=\mu\essinf_{\theta\in\Theta}U(\theta)\implies \inf_{m\in\calM_0}\E_F[m(\theta)U(\theta)] = \mu\essinf_{\theta\in\Theta}U(\theta).\]
Thus, the regulator's problem can be written as
\[\adjustlimits\sup_{q\in\calQ_{[0,1]}}\sup_{t(\und\theta)\in\R_+}\left\{\mu\left[\und\theta q(\und\theta)-t(\und\theta)\right]: \int_\Theta q(\theta)\,\dd F(\theta)\leq Q\right\}.\]
Because wait times must be nonnegative, $t(\und\theta)\geq 0$.  Moreover, since $q$ is nondecreasing,
\[q(\und\theta)\leq \int_{\Theta}q(\theta)\,\dd F(\theta) \leq Q.\]
Therefore, the value of the regulator's problem \eqref{eq:P_costly} is bounded from above by 
\[\adjustlimits\sup_{q\in\calQ_{[0,1]}}\sup_{t(\und\theta)\in\R_+}\left\{\mu\left[\und\theta q(\und\theta)-t(\und\theta)\right]: \int_\Theta q(\theta)\,\dd F(\theta)\leq Q\right\}\leq \mu\und\theta Q.\]
It is easy to see that the mechanism $(q^*,t^*)=(Q,0)$ is incentive-compatible and attains this value.  Hence, $(q^*,t^*)$ is optimal, as claimed.  This solution is unique up to $F$-null sets whenever $\und\theta>0$.
\end{proof}

Under the optimal mechanism, the regulator randomly allocates the good without requiring agents to engage in costly screening.  It can be implemented by a symmetric lottery under which each agent receives the good with probability $Q$.

The result provides a robustness-based rationale for institutional arrangements that deliberately avoid costly screening, even in settings where such screening could improve allocative efficiency under stronger informational assumptions.  In the context of dispute resolution, for instance, arbitration and settlement procedures can be understood as responses to uncertainty about how litigation costs correlate with claim strength.  Another example is the increasing popularity of lotteries---rather than ordeals \citep{nicholszeckhauser82}---in the allocation of public housing, vouchers, school seats, and other public services.  Whereas such lotteries are often justified on equity or administrative grounds, \Cref{prop:costly} shows that they are also consistent with maximizing worst-case total surplus.

From a technical perspective, \Cref{prop:costly} shows that the robustness logic of \Cref{thm:unknown} extends to environments in which the main friction affects agents' payoffs directly (here, through waiting costs), rather than entering only through the regulator's objective via the allocation rule.  As the proof shows, the paper's concentration argument can be applied directly to utility rather than to the allocation rule itself.

\clearpage
\putbib[master_bibliography]
\clearpage

\end{bibunit}

\end{document}

%% file: working_arxiv.tex
\usepackage{amssymb,amsthm,mathtools,bbm}
\usepackage{enumitem,cite}
\usepackage[longnamesfirst]{natbib}
\usepackage{multirow,abstract}
\usepackage{framed,mdframed,cancel,setspace}
\usepackage[bottom,hang]{footmisc}
\usepackage{graphicx,tikz}
\usepackage[colorlinks=true,allcolors=NavyBlue,citecolor=NavyBlue]{hyperref}
\hypersetup{breaklinks=true}
\usetikzlibrary{calc,positioning}
\usepackage[font={small,it}]{caption}
\usepackage[noabbrev]{cleveref}
\usepackage[T1]{fontenc}

\usepackage{breakurl}

\newtheoremstyle{ans}{15pt}{20pt}{}{}{\bfseries}{}{ }{\thmname{#1}\thmnote{ #3}.}
\theoremstyle{ans}
\newtheorem*{soln}{Solution}

\newcounter{lemno}
\newcounter{app_lemno}
\newcounter{propno}
\newcounter{app_propno}
\newcounter{thmno}
\newcounter{rethmno}
\newcounter{clmno}
\newcounter{app_thmno}
\newcounter{defno}
\newcounter{factno}
\newcounter{corno}
\newcounter{assno}
\newcounter{reassno}
\newcounter{examplenum}
\newcounter{remarkno}
\newmdenv[backgroundcolor=blue!7,linewidth=0pt]{bBox}
\newmdenv[backgroundcolor=green!7,linewidth=0pt]{gBox}
\newmdenv[backgroundcolor=red!7,linewidth=0pt]{rBox}
\newmdenv[backgroundcolor=gray!7,linewidth=0pt]{grayBox}

\newtheoremstyle{statementStyle}
{12pt}
{12pt}
{}
{}
{}
{{\bf .}\;}
{0.25em}
{{\bf\thmname{#1}\thmnumber{ #2}}
\thmnote{ (#3)}} 
\makeatother

\theoremstyle{statementStyle}

\newtheoremstyle{statementStyle}
{12pt}
{12pt}
{\it}
{}
{}
{{\bf .}\;}
{0.25em}
{{\bf\thmname{#1}\thmnumber{ #2}}
\thmnote{ (#3)}} 
\makeatother

\theoremstyle{plain}
\newtheorem{lemma}[lemno]{Lemma}

\newtheorem{theorem}[thmno]{Theorem}

\newtheorem{coro}[corno]{Corollary}
\newtheorem{exerciseT}[thmno]{Exercise}
\newtheorem{exampleT}[examplenum]{Example}
\newtheorem{defn}[defno]{Definition}

\Crefname{assumption}{Assumption}{Assumptions}
\newtheorem{proposition}[propno]{Proposition}
\Crefname{proposition}{Proposition}{Propositions}

\Crefname{remark}{Remark}{Remarks}

\Crefname{claim}{Claim}{Claims}

\newtheoremstyle{nbstatementStyle}
{12pt}
{12pt}
{\it}
{}
{}
{{\bf .}\;}
{0.25em}
{{\bf\thmname{#1}}
\thmnote{ (#3)}} 
\makeatother

\theoremstyle{nbstatementStyle}

\newtheoremstyle{restatementStyle}
{12pt}
{12pt}
{\it}
{}
{\bfseries}
{{\bf .}\;}
{0.25em}
{\thmname{#1}\thmnumber{ #2'}
\thmnote{ (#3)}} 
\makeatother

\theoremstyle{restatementStyle}

\newtheorem{retheorem}[rethmno]{Theorem}
\Crefname{retheorem}{Theorem}{Theorems}

\newtheorem{prob}{Problem}

\newtheorem{conj}{Conjecture}
\newtheorem{prope}{Property}

\newtheoremstyle{reexerciseStyle}
{12pt}
{12pt}
{}
{}
{\bfseries}
{{\bf .}\;}
{0.25em}
{\thmname{#1}\thmnumber{ #2}
\thmnote{ (#3)}} 
\makeatother

\newtheoremstyle{exerciseStyle}
{12pt}
{12pt}
{}
{}
{\bfseries}
{{\bf .}\;}
{0.25em}
{\thmname{#1}\thmnumber{ #2}
\thmnote{ (#3)}} 
\makeatother

\theoremstyle{exerciseStyle}

\newtheoremstyle{teach}{15pt}{\topsep}{}{}{\itshape}{}{ }{\thmname{#1}\thmnote{ #3}.}
\theoremstyle{teach}
\newtheorem*{prooflemma}{Proof of lemma}

\newtheoremstyle{rem}{5pt}{0pt}{\color{black}}{}{\bfseries}{}{ }{\thmname{#1}\thmnote{ #3}.}
\theoremstyle{rem}

\definecolor{darkgreen}{rgb}{0.0, 0.75, 0.2}

\makeatletter
\renewcommand*\env@matrix[1][*\c@MaxMatrixCols c]{%
  \hskip -\arraycolsep
  \let\@ifnextchar\new@ifnextchar
  \array{#1}}
\makeatother

\newcommand{\dd}{\mathrm{d}}

\newcommand{\la}{\lambda}

\newcommand{\e}{\varepsilon}

\newcommand{\R}{\mathbb{R}}

\newcommand{\bone}{\boldsymbol{1}}

\newcommand{\BAR}{\overline}

\newcommand{\und}[1]{\underline{#1}}

\newcommand{\AND}{\quad\text{and}\quad}

\renewcommand{\(}{\left(}
\renewcommand{\)}{\right)}

\newcommand{\eg}{e.g.}
\newcommand{\ie}{i.e.}

\DeclareMathOperator*{\argmax}{arg\,max}
\DeclareMathOperator*{\argmin}{arg\,min}
\DeclareMathOperator{\E}{\mathbf{E}}

\def\NN{\mathbb{N}}

\def\calM{\mathcal{M}}

\def\calQ{\mathcal{Q}}

\usepackage{datetime}

\newdateformat{monthyeardate}{%
  \monthname[\THEMONTH] \THEYEAR}
\newdateformat{daymonthyeardate}{%
  \monthname[\THEMONTH] \THEDAY, \THEYEAR}

\title{\LARGE\hspace*{\fill}\\[1ex]\papertitle\\[2ex]}
\author{\large\name\\[1ex] \affiliation}
\date{\monthyeardate\today}

\emergencystretch=\maxdimen
\makeatletter
\def\@xfootnote[#1]{%
  \protected@xdef\@thefnmark{#1}%
  \@footnotemark\@footnotetext}
\makeatother

\makeatletter
\ifFN@para
\else
  \long\def\@makefntext#1{%
    \ifFN@hangfoot
      \bgroup
      \setbox\@tempboxa\hbox{%
        \ifdim\footnotemargin>0pt
          \hb@xt@\footnotemargin{\@makefnmark\hss}%
        \else
          \@makefnmark\hskip-\footnotemargin      
        \fi
      }%
      \leftmargin\wd\@tempboxa
      \rightmargin\z@
      \linewidth \columnwidth
      \advance \linewidth -\leftmargin
      \parshape \@ne \leftmargin \linewidth
      \footnotesize
      \@setpar{{\@@par}}%
      \leavevmode
      \llap{\box\@tempboxa}%
      \parskip\hangfootparskip\relax
      \parindent\hangfootparindent\relax
    \else
      \parindent1em
      \noindent
      \ifdim\footnotemargin>\z@
        \hb@xt@ \footnotemargin{\hss\@makefnmark}%
      \else
        \ifdim\footnotemargin=\z@
          \llap{\@makefnmark}%
        \else
          \llap{\hb@xt@ -\footnotemargin{\@makefnmark\hss}}%
        \fi
      \fi
    \fi
    \footnotelayout#1%
    \ifFN@hangfoot
      \par\egroup
    \fi
  }
\fi
\makeatother

\makeatletter
\def\@endtheorem{\endtrivlist}
\makeatother

\newcommand{\citepos}[1]{\citeauthor{#1}'s (\citeyear{#1})}